\documentclass[letterpaper,singlecolumn,9pt]{sig-alternate-2013}

\usepackage[T1]{fontenc}
\usepackage[utf8]{inputenc}
\usepackage{t1enc}

\usepackage{amsmath}
\usepackage{amsfonts}
\usepackage{amssymb}

\usepackage{multirow,bigdelim}
\usepackage{rotating}
\usepackage{caption}
\usepackage[hang]{subfigure}
\usepackage{algorithm}
\usepackage[noend]{algpseudocode}
\usepackage[]{units}

\usepackage{tikz}
\usepackage{gnuplot-lua-tikz}
\usepackage{cite}
\usepackage{url}
\usepackage{xspace}
\usepackage{multicol}
\usepackage{ifpdf}

\newtheorem{proposition}{Proposition}
\newtheorem{theorem}{Theorem}
\newtheorem{definition}{Definition}
\newtheorem{lemma}{Lemma}

\usepackage{bigfoot}

\newcommand{\abs}[1]{\lvert #1 \rvert}

\DeclareMathOperator{\polylog}{polylog}

\DeclareMathOperator{\last}{last}
\DeclareMathOperator{\lp}{lp}

\DeclareMathOperator{\Delete}{delete}

\DeclareMathOperator{\level}{level}
\DeclareMathOperator{\xbwl}{xbwb}
\DeclareMathOperator{\pDAG}{pDAG}

\DeclareMathOperator{\access}{access}
\DeclareMathOperator{\rank}{rank}
\DeclareMathOperator{\select}{select}
\DeclareMathOperator{\getbits}{bits}
\DeclareMathOperator{\leafpush}{leaf\_push}

\newcommand{\XBWL}{\emph{XBW-b}\xspace}
\newcommand{\stind}{\ensuremath{\mathcal{S}}}

\newcommand{\Left}{\ensuremath{\text{left}}}
\newcommand{\Right}{\ensuremath{\text{right}}}

\newcommand{\id}{\ensuremath{\text{id}}}
\newcommand{\Root}{\ensuremath{\text{root}}}

\newcommand{\Copy}{\ensuremath{\text{copy}}}
\newcommand{\New}{\ensuremath{\text{new\_node}}}

\newcommand{\lambertW}{\ensuremath{\mathcal{W}}}

\newcommand{\lookup}{\ensuremath{\texttt{lookup}}}

\newcommand{\triefold}{\ensuremath{\texttt{trie\_fold}}}
\newcommand{\update}{\ensuremath{\texttt{update}}}
\newcommand{\compress}{\ensuremath{\texttt{compress}}}
\newcommand{\decompress}{\ensuremath{\texttt{decompress}}}
\newcommand{\Put}{\ensuremath{\texttt{put}}}
\newcommand{\Get}{\ensuremath{\texttt{get}}}

\DeclareNewFootnote[para]{default}
\DeclareNewFootnote[para]{E}[alph]
\MakeSortedPerPage{footnoteE}
\allowdisplaybreaks[4]

\newcommand*{\titlePage}{\begingroup 

\onecolumn
\centering 

\begin{center}
  \includegraphics[width=200pt]{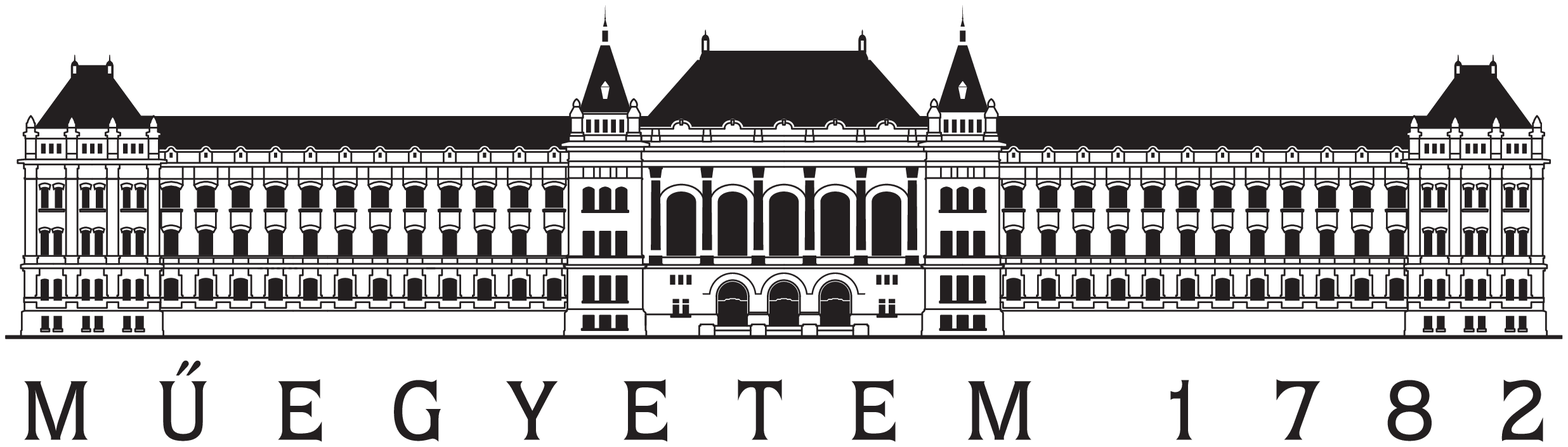}
\end{center}

\vspace*{6\baselineskip} 

{\Huge Compressing IP Forwarding
     Tables:\\[0.3\baselineskip]Towards Entropy Bounds and
     Beyond}\\[0.2\baselineskip] 

\vspace*{\baselineskip} 

{\Huge A revised technical report}\\[0.2\baselineskip]

\scshape 

\vspace*{4\baselineskip} 

Manuscript originally appeared as\\ 
\vspace*{\baselineskip} 

\slshape 
G.~R\'etv\'ari, J.~Tapolcai, A.~K\H{o}r\"osi, A.~Majd\'an, and Z.~Heszberger,\\
  ``Compressing {IP} forwarding tables: Towards entropy bounds and beyond,''\\
  ACM SIGCOMM, 2013.\\

\scshape 
\vspace*{\baselineskip} 
This technical report is a revised version of that manuscript,\\containing a
number of important corrections\\ to the original text.\\[\baselineskip]

\vspace*{10\baselineskip} 

Authors \\[\baselineskip]
{\Large G\'abor~R\'etv\'ari\\J\'anos~Tapolcai\\Attila~K\H{o}r\"osi\\Andr\'as
  Majd\'an\\Zal\'an Heszberger\par} 
\vspace*{\baselineskip} 
{\itshape Department of Telecommunications and Media Informatics\\
  Budapest University of Technology and Economics \par} 

\vspace*{1\baselineskip}

Budapest, Hungary, 2014\par 

\vfill 

\endgroup}

\begin{document}

\conferenceinfo{SIGCOMM'13,} {August 12--16, 2013, Hong Kong, China.} 
 \CopyrightYear{2013} 
 \crdata{978-1-4503-2056-6/13/08} 
 \clubpenalty=10000 
 \widowpenalty = 10000

\titlePage

\title{Compressing IP Forwarding Tables:\\Towards Entropy Bounds and Beyond}

\numberofauthors{1}
\author{
  \alignauthor
  G\'abor~R\'etv\'ari, J\'anos~Tapolcai, Attila~K\H{o}r\"osi, Andr\'as
  Majd\'an, Zal\'an Heszberger\\
  \affaddr{Department of Telecommunications and Media Informatics\\
    Budapest University of Technology and Economics
    \email{\{retvari,tapolcai,korosi,majdan,heszi\}@tmit.bme.hu}%
  }%
}%

\maketitle

\begin{abstract}
  Lately, there has been an upsurge of interest in compressed data
  structures, aiming to pack ever larger quantities of information into
  constrained memory without sacrificing the efficiency of standard
  operations, like random access, search, or update.  The main goal of this
  paper is to demonstrate how data compression can benefit the networking
  community, by showing how to squeeze the IP Forwarding Information Base
  (FIB), the giant table consulted by IP routers to make forwarding
  decisions, into information-theoretical entropy bounds, with essentially
  zero cost on longest prefix match and FIB update.  First, we adopt the
  state-of-the-art in compressed data structures, yielding a static
  entropy-compressed FIB representation with asymptotically optimal lookup.
  Then, we re-design the venerable prefix tree, used commonly for IP lookup
  for at least 20 years in IP routers, to also admit entropy bounds and
  support lookup in optimal time and update in nearly optimal time.
  Evaluations on a Linux kernel prototype indicate that our compressors
  encode a FIB comprising more than 440K prefixes to just about 100--400
  KBytes of memory, with a threefold increase in lookup throughput and no
  penalty on FIB updates.
\end{abstract}





\section{Introduction}
\label{sec:intro}

Data compression is widely used in processing large volumes of information.
Not just that convenient compression tools are available to curtail the
memory footprint of basically any type of data, but these tools also come
with 
theoretical guarantee that the compressed size is indeed minimal, in terms of
some suitable notion of \emph{entropy} \cite{Cover:1991:EIT:129837}.
Correspondingly, data compression has found its use in basically all aspects
of computation and networking practice, ranging from text or multimedia
compression \cite{Ziviani:2000:CKN:619057.621588} to the very heart of
communications protocols \cite{rfc1979} and operating systems \cite{zfs}.

Traditional 
compression algorithms 
do not admit standard queries, like pattern matching or random access, right
on the compressed form, which severely hinders their applicability.  An
evident workaround is to decompress the data prior to accessing it, but this
pretty much defeats the whole purpose.  The alternative is to maintain a
separate index dedicated solely to navigate the content, but the sheer size
of the index can become prohibitive in many cases
\cite{Hon:2010:CIR:1875737.1875761, Navarro:2007:CFI:1216370.1216372}.

It is no surprise, therefore, that the discovery of compressed self-indexes
(or, within the context of this paper, \emph{compressed data structures})
came as a real breakthrough \cite{Ferragina:2000:ODS:795666.796543}.  A
compressed data structure is, loosely speaking, an entropy-sized index on
some data that allows the complete recovery of the original content as well
as fast queries on it \cite{Hon:2010:CIR:1875737.1875761,
  Navarro:2007:CFI:1216370.1216372, Ferragina:2000:ODS:795666.796543, libcds,
  Makinen:2008:DES:1367064.1367072, Ziviani:2000:CKN:619057.621588,
  SilvadeMoura:2000:FFW:348751.348754, Raman:2002:SID:545381.545411,
  Ferragina:2007:CRS:1240233.1240243}.  What is more, as the compressed form
occupies much smaller space than the original representation, and hence is
more friendly to CPU cache, the time required to answer a query is often far
less than if the data had not been compressed
\cite{Ziviani:2000:CKN:619057.621588, SilvadeMoura:2000:FFW:348751.348754}.
\emph{Compressed data structures, therefore, turn out one of the rare cases
  in computer science where there is no space-time trade-off}.

Researchers and practitioners working with big data were quick to recognize
this win-win situation and came up with compressed self-indexes, and
accompanying software tools, for a broad set of applications; from
compressors for sequential data like bitmaps (\texttt{RRR}),
\cite{Raman:2002:SID:545381.545411}) and
text documents (\texttt{CGlimpse} \cite{Ferragina:2000:ODS:795666.796543},
wavelet trees \cite{Ferragina:2007:CRS:1240233.1240243}); compression
frontends to information retrieval systems 
(\texttt{MG4J} \cite{citeulike:2046731}, \texttt{LuceneTransform}
\cite{luctrans}) and dictionaries (\texttt{MG} \cite{WittenMoffatBell99}); to
specialized tools for structured data, like XML\slash HTML\slash DOM
(\texttt{XGRIND} \cite{Tolani02xgrind:a}, \texttt{XBZIPINDEX}
\cite{Ferragina:2009:CIL:1613676.1613680}),
graphs (\texttt{WebGraph} \cite{webgraph}), 3D models (\texttt{Edgebreaker}
\cite{Rossignac99edgebreaker:connectivity}), genomes and protein sequences
(\texttt{COMRi} \cite{Sun:2003:CCM:937976.938085}), multimedia, source and
binary program code, formal grammars, etc. \cite{WittenMoffatBell99}.  With
the advent of replacements for the standard file compression tools
(\texttt{LZgrep} \cite{Navarro:2005:LBS:1087482.1087483}) and generic
libraries (\texttt{libcds} \cite{libcds}), we might be right at the verge of
seeing compressed data structures go mainstream.

Curiously, this revolutionary change has gone mostly unnoticed in the
networking community, even though this field is just one of those affected
critically by 
skyrocketing volumes of data.  A salient example of this trend is the case of
the IP Forwarding Information Base (FIB), used by Internet routers to make
forwarding decisions, which has been literally deluged by the rapid growth of
the routed IP address space. 
Consequently, there has been a flurry of activity to find space-efficient FIB
representations \cite{1200040, 253403, 772439, 752164,
  Bando:2012:FBI:2369183.2369204, DBLP:conf/infocom/GuptaPB00,
  Degermark:1997:SFT:263105.263133, Dharmapurikar:2003:LPM:863955.863979,
  draves:99, Eatherton:2004:TBH:997150.997160, Han:2010:PGS:1851182.1851207,
  Hasan:2005:DPM:1080091.1080116, Ioannidis:2005:LCD:1114718.1648670,
  Liu:2012:EFC:2427036.2427039, Sklower91atree-based,
  Song:2005:SST:1099544.1100365, Srinivasan:1998:FIL:277858.277863,
  Uzmi:2011:SPN:2079296.2079325, Zec:2012:DTB:2378956.2378961,
  Waldvogel:1997:SHS:263105.263136}, yet very few of these go beyond ad-hoc
schemes and compress to information-theoretic limits, let alone come with a
convenient notion of FIB entropy. Taking the example of compressing IP FIBs
as a use case, thus, our aim in this paper is to popularize compressed data
structures to the networking community.

\subsection{FIB Compression}
\label{sec:fib-comp}

There are hardly any data structures in networking affected as compellingly
by the growth of the Internet as the IP FIB.  Stored in the line card memory
of routers, the FIB maintains an association from every routed IP prefix to
the corresponding next-hop, and it is queried on a packet-by-packet basis at
line speed (in fact, it is queried \emph{twice} per packet, considering
reverse path forwarding check).  Lookup in FIBs is not trivial either, as
IP's longest prefix match rule requires the most specific entry to be found
for each destination address.  Moreover, as Internet routers operate in an
increasingly dynamic environment \cite{Elmokashfi:2012:BCE:2317330.2317349},
the FIB needs to support hundreds of updates to its content each second.

As of 2013, the number of active IPv4 prefixes in the Default Free Zone is
more than 440,000 and counting, and IPv6 quickly follows suit \cite{potaroo}.
Correspondingly, FIBs continue to expand both in size and management burden.
As a quick reality check, the Linux kernel's \texttt{fib\_trie} data
structure \cite{772439}, when filled with this many prefixes, occupies tens
of Mbytes of memory, takes several minutes to download to the forwarding
plane, and is still heavily debated to scale to multi-gigabit speeds
\cite{1709673}.  Commercial routers suffer similar troubles, aggravated by
the fact that line card hardware is more difficult to upgrade than software
routers.

Many have argued that mounting FIB memory tax will, sooner or later, present
a crucial data-plane performance bottleneck for IP routers \cite{rfc4984}.
But even if the scalability barrier will not prove impenetrable
\cite{godfrey-smaller}, the growth of the IP forwarding table still poses
compelling difficulties.  Adding further fast memory to line cards boosts
silicon footprint, heat production, and power budget, all in all, the
CAPEX\slash OPEX associated with IP network gear, and forces operators into
rapid upgrade cycles \cite{Zhao:2010:RSO:1878170.1878174,
  Khare:2010:ETG:1878170.1878183}.  Large FIBs also complicate maintaining
multiple virtual router instances, each with its own FIB, on the same
physical hardware \cite{DBLP:conf/icnp/SongKHL09} and build up huge control
plane to data plane delay for FIB resets
\cite{Francois:2005:ASI:1070873.1070877}. 

Several recent studies have identified \emph{FIB aggregation} as an effective
way to reduce FIB size, extending the lifetime of legacy network devices and
mitigating the Internet routing scalability problem, at least temporarily
\cite{Zhao:2010:RSO:1878170.1878174, Khare:2010:ETG:1878170.1878183}.  FIB
aggregation is a technique to transform some initial FIB representation into
an alternative form that, supposedly, occupies smaller space but still
provides fast lookup.  Recent years have seen an impressive reduction in FIB
size: from the initial 24 bytes\slash prefix (prefix trees
\cite{Sklower91atree-based}), use of hash-based schemes
\cite{Waldvogel:1997:SHS:263105.263136, Bando:2012:FBI:2369183.2369204},
path- and level-compressed multibit tries
\cite{Srinivasan:1998:FIL:277858.277863, 772439, 752164}, tree-bitmaps
\cite{Eatherton:2004:TBH:997150.997160}, etc., have reduced FIB memory tax to
just about 2--4.5 bytes/prefix \cite{Degermark:1997:SFT:263105.263133,
  Zec:2012:DTB:2378956.2378961, Uzmi:2011:SPN:2079296.2079325}.  Meanwhile,
lookup performance has also improved \cite{772439}.

The evident questions ``Is there an ultimate limit in FIB aggregation?''  and
``Can FIBs be reduced to fit in fast ASIC SRAM\slash CPU cache entirely?''
have been asked several times before \cite{752164,
  Srinivasan:1998:FIL:277858.277863, draves:99,
  Degermark:1997:SFT:263105.263133}.  In order to answer these questions, we
need to go beyond conventional FIB aggregation to find \emph{new compressed
  FIB data structures that encode to entropy-bounded space and support lookup
  and update in optimal time}.
We coined the term \emph{FIB compression} to mark this ambitious undertaking
\cite{Retvari:2012:CIF:2390231.2390232}.  Accordingly, this paper is
dedicated to the theory and practice of FIB compression.

\subsection{Our Contributions}
\label{sec:contrib}

Our contributions on FIB compression are two-fold: based on the labeled tree
entropy measure of Ferragina \emph{et al.}
\cite{Ferragina:2009:CIL:1613676.1613680} we specify a compressibility metric
called FIB entropy, then we propose two entropy-compressed FIB data
structures.

Our first FIB encoder, \XBWL, is a direct application of the state-of-the-art
in compressed data structures to the case of IP FIBs.  \XBWL compresses a
contemporary FIB to the entropy limit of just $100$--$300$ Kbytes and, at the
same time, provides longest prefix match in asymptotically optimal time.
Unfortunately, it turns out that the relatively immature hardware and
software background for compressed string indexes greatly constrain the
lookup and update performance of \XBWL.  Therefore, we also present a
practical FIB compression scheme, called the trie-folding algorithm.

\emph{Trie-folding} is in essence a ``compressed'' reinvention of prefix
trees, a commonly used FIB implementation in IP routers, and therefore
readily deployable with minimal or no modification to router ASICs
\cite{ezchip}.  We show that trie-folding compresses to within a small
constant factor of FIB entropy, supports lookup in strictly optimal time, and
admits updates in nearly optimal time for FIBs of reasonable entropy (see
later for precise definitions). The viability of trie-folding will be
demonstrated on a Linux prototype and an FPGA implementation.  By extensive
tests on real and synthetic IP FIBs, we show that trie-folding supports tens
of millions of IP lookups and hundreds of thousands updates per second, in
less than $150$--$500$ Kbytes of memory.

\subsection{Structure of the Paper}
\label{sec:struct}

The rest of the paper is organized as follows.  In the next section, we
survey standard FIB representation schemes and cast compressibility metrics.
In Section~\ref{sec:mbw} we describe \XBWL, while in
Section~\ref{sec:trie-fold} we introduce trie-folding and we establish
storage size bounds.  Section~\ref{sec:num-eval} is devoted to numerical
evaluations and measurement results, Section~\ref{sec:related} surveys
related literature, and finally Section~\ref{sec:conc} concludes the paper.

\section{Prefix Trees and Space Bounds}
\label{sec:redundancy}

Consider the sample IP routing table in Fig.~\ref{tab:textual}, storing
address-prefix-to-next-hop associations in the form of an index into a
neighbor table, which maintains neighbor specific information, like next-hop
IP address, aliases, ARP info, etc.  Associate a unique \emph{label}, taken
from the alphabet $\Sigma$, with each next-hop in the neighbor table.  We
shall usually treat labels as positive integers, complemented with a special
\emph{invalid label} $\perp \in \Sigma$ to mark blackhole routes.  Let $N$
denote the number of entries in the FIB and let $\delta =\abs{\Sigma}$ be the
number of next-hops.  An IP router does not keep an adjacency with every
other router in the Internet, thus $\delta \ll N$. Specifically, we assume
that $\delta$ is $O(\polylog N)$ or $O(1)$, which is in line with reality
\cite{Teixeira:2003:SPD:948205.948247, 5928930}.  Finally, let $W$ denote the
width of the address space in bits (e.g., $W=32$ for IPv4).

\begin{figure}
  \centering
  \subtable[]{%
    \begin{small}%
      \begin{tabular}[t]{r|c}%
        prefix & label\\
        \hline
        -/0 & 2\\
        0/1 & 3\\
        00/2 & 3\\
        001/3 & 2\\
        01/2 & 2\\
        011/3 & 1\\
      \end{tabular}%
    \end{small}%
\label{tab:textual}}%
  \subfigure[]{%
    \begin{tikzpicture}
      [baseline=(root.base),scale=.22, minimum size=8,inner sep=1pt,
      node/.style={anchor=center,circle,draw=black!60,font=\scriptsize}] {
        \node (0_2) at (1.6,0) [node] {$2$};%
        \node (1_2) at (4.8,0) [node] {$1$};%

        \node (i_0) at (.8,3) [node] {$3$}; %
        \path[->,>=latex] (i_0) 
        edge
        node[above right,minimum size=1ex,scale=.7]{$1$}%
        (0_2);%
        \node (i_1) at (4,3) [node] {$2$}; %
        \path[->,>=latex] (i_1) 
        edge
        node[above right,minimum size=1ex,scale=.7]{$1$}%
        (1_2);%

        \node (i2_0) at (2.4,6) [node] {$3$}; %
        \path[->,>=latex] (i2_0) edge
        node[above left,minimum size=1ex,scale=.7]{$0$}%
        (i_0) edge
        node[above right,minimum size=1ex,scale=.7]{$1$}%
        (i_1); %

        \node (root) at (4,9) [node] {$2$}; %
        \path[->,>=latex] (root) edge
        node[above left,minimum size=1ex,scale=.7]{$0$}%
        (i2_0) 
        ;%
      };
    \end{tikzpicture}\label{fig:pretix-trie}}%
  \hskip5pt%
  \subfigure[]{%
    \begin{tikzpicture}
      [baseline=(root.base),scale=.22, minimum size=8,inner sep=1pt,
      node/.style={anchor=center,circle,draw=black!60,font=\scriptsize}] {
        \node (0_1) at (0,0) [node] {$3$};%
        \node (1_2) at (4.8,0) [node] {$1$};%
        \node (i_0) at (.8,3) [node] {}; %
        \path[->,>=latex] (i_0) edge (0_1) 
        ;%
        \node (i_1) at (4,3) [node] {}; %
        \path[->,>=latex] (i_1) 
        edge (1_2);%

        \node (i2_0) at (2.4,6) [node] {}; %
        \path[->,>=latex] (i2_0) edge (i_0) edge (i_1); %

        \node (root) at (4,9) [node] {$2$}; %
        \path[->,>=latex] (root) edge (i2_0) 
        ;%
      };
    \end{tikzpicture}\label{fig:ortc}}%
  \hskip3pt%
  \subfigure[]{%
    \begin{tikzpicture}
      [baseline=(root.base),scale=.22, minimum size=8,inner sep=1pt,
      node/.style={anchor=center,circle,draw=black!60,font=\scriptsize}] {
        \node (0_1) at (0,2.5) [node] {$3$};%
        \node (0_2) at (1.6,2.5) [node] {$2$};%
        \node (1_1) at (3.2,2.5) [node] {$2$};%
        \node (1_2) at (4.8,2.5) [node] {$1$};%

        \node (i2_0) at (2.4,6) [node] {}; %
        \path[->,>=latex] (i2_0) edge (0_1) edge (0_2) edge (1_1) edge (1_2); %
        \node (i2_1) at (5.6,6) [node] {$2$}; %

        \node (root) at (4,9) [node] {}; %
        \path[->,>=latex] (root) edge (i2_0) edge (i2_1);%
        \node (a) at (0, 0.3) {};
      };
    \end{tikzpicture}\label{fig:multibit}}%
  \subfigure[]{%
    \begin{tikzpicture}
      [baseline=(root.base),scale=.22, minimum size=8,inner sep=1pt,
      node/.style={anchor=center,circle,draw=black!60,font=\scriptsize}] {
        \node (0_1) at (0,0) [node] {$3$};%
        \node (0_2) at (1.6,0) [node] {$2$};%
        \node (1_1) at (3.2,0) [node] {$2$};%
        \node (1_2) at (4.8,0) [node] {$1$};%

        \node (i_0) at (.8,3) [node] {}; %
        \path[->,>=latex] (i_0) edge (0_1) edge (0_2);%
        \node (i_1) at (4,3) [node] {}; %
        \path[->,>=latex] (i_1) edge (1_1) edge (1_2);%

        \node (i2_0) at (2.4,6) [node] {}; %
        \path[->,>=latex] (i2_0) edge (i_0) edge (i_1); %
        \node (i2_1) at (5.6,6) [node] {$2$}; %

        \node (root) at (4,9) [node] {}; %
        \path[->,>=latex] (root) edge (i2_0) edge (i2_1);%
      };
    \end{tikzpicture}\label{fig:leaf-pushed}}%



  \caption{Representations of an IP forwarding table: tabular form with
    address in binary format, prefix length and next-hop address label (a);
    prefix tree with state transitions marked (b); ORTC-compressed prefix
    tree (c); level-compressed multibit trie (d); and leaf-pushed trie (e).}
\label{fig:fibs}
\end{figure}
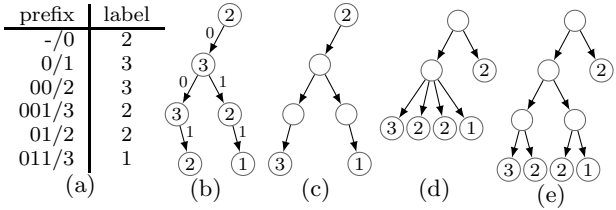

To actually forward a packet, we need to find the entry that matches the
destination address in the packet on the greatest number of bits, starting
from the MSB.  For the address $0111$, each of the entries $-/0$ (the default
route), $0/1$, $01/2$, and $011/3$ match.  As the most specific entry is the
last one, the lookup operation yields the next-hop label $1$.  This is then
used as an index into the neighbor table and the packet is forwarded on the
interface facing that neighbor.  This tabular representation is not really
efficient, as a lookup or update operation requires looping through each
entry, taking $O(N)$ time. The storage size is $(W + \lg \delta)N$
bits\footnote{The notation $\lg x$ is shorthand for $\lceil \log_2(x)
  \rceil$.}.

Binary prefix trees, or \emph{tries} \cite{Sklower91atree-based}, support
lookup and update much more efficiently (see Fig.~\ref{fig:pretix-trie}).  A
trie is a labeled ordinal tree, in which every path from the root node to a
leaf corresponds to an IP prefix and lookup is based on successive bits of
the destination address: if the next bit is $0$ proceed to the left sub-trie,
otherwise proceed to the right, and if the corresponding child is missing
return the last label encountered along the way.  Prefix trees generally
improve the time to perform a lookup or update from linear to $O(W)$ steps,
although memory size increases somewhat.

A prefix tree representation is usually not unique, which opens the door to a
wide range of optimization strategies to find more space-efficient forms.
For instance, the prefix tree in Fig.~\ref{fig:ortc} is \emph{forwarding
  equivalent} with the one in Fig.~\ref{fig:pretix-trie}, in that it orders
the same label to every complete $W$ bit long key, yet contains only $3$
labeled nodes instead of $7$ (see the ORTC algorithm in \cite{draves:99,
  Uzmi:2011:SPN:2079296.2079325}).  Alternatively, \emph{level-compression}
\cite{Srinivasan:1998:FIL:277858.277863, 772439, 752164} is a technique to
remove excess levels from a binary trie to obtain a forwarding equivalent
multibit trie that is substantially smaller (see Fig.~\ref{fig:multibit}).

A standard technique to obtain \emph{a unique, normalized form} of a prefix
tree is \emph{leaf-pushing} \cite{draves:99,
  Srinivasan:1998:FIL:277858.277863, DBLP:conf/icnp/SongKHL09}: in a first
preorder traversal labels are pushed from the parents towards the children,
and then in a second postorder traversal each parent with identically labeled
leaves is substituted with a leaf marked with the children's label (see
Fig.~\ref{fig:leaf-pushed}).  The resultant trie is called a
\emph{leaf-labeled} trie since interior nodes no longer maintain labels, and
it is also a \emph{proper} binary trie with nice structure: any node is
either a leaf node or it is an interior node with exactly two children.
Updates to a leaf-pushed trie, however, may be expensive; modifying the
default route, for instance, can result in practically all leaves being
relabeled, taking $O(N)$ steps in the worst-case.

\subsection{Information-theoretic Limit}
\label{sec:inform-theor-redund}

How can we know for sure that a particular prefix tree representation, from
the many, is indeed space-efficient?  To answer this question, we need
information-theoretically justified storage size bounds.

The first verifiable cornerstone of a space-efficient data structure is
whether its size meets the \emph{information-theoretic lower bound},
corresponding to the minimum number of bits needed to uniquely identify any
instance of the data.  For example, there are exactly $\delta^n$ strings of
length $n$ on an alphabet of size $\delta$, and to be able to distinguish
between any two we need at least $\lg(\delta^n) \approxeq n \lg \delta$ bits. In this
example even a 
naive string representation meets the bound, but in more complex cases
attaining it is much more difficult.

This argumentation generalizes from strings to leaf-labeled tries as follows
(see also Ferragina \emph{et al.} \cite{Ferragina:2009:CIL:1613676.1613680}).
\begin{proposition}\label{prop:inform-theor-bound}
  Let $T$ be a proper, binary, leaf-labeled trie with $n$ leaves on an
  alphabet of size $\delta$.  The information-theoretic lower bound to encode
  $T$ is $2n + n \lg \delta$ bits\footnoteE{{\it Erratum:} In the original
    manuscript \cite{sigcomm_2013} the information-theoretic lower bound is
    wrongly set to $4n + n \lg \delta$. See the note below for the
    explanation.}.
\end{proposition}

The bound is easily justified with a simple counting argument.  The number of
proper binary trees on $n$ leaves is the $(n-1)$-th Catalan number
$C_{n-1}=\tfrac1{n}\binom{2n-2}{n-1}$, therefore we need at least $\lg
C_{n-1} = 2n - \Theta(\log n)$ bits\footnoteE{{\it Erratum:} The
  information-theoretic lower bound and the entropy are off by a constant
  factor $2$ in \cite{sigcomm_2013}.  The reason is that the original version
  takes the number of trees on $n$ \emph{nodes} instead of $n$ \emph{leaves},
  thus it wrongly puts the number of bits to encode the tree to $\approx 4n$
  bits.} to encode the tree itself \cite{63533}; storing the label map
defined on the $n$ leaves of $T$ requires an additional $n \lg \delta$ bits;
and assuming that the two are independent we need $2n + n \lg \delta$ bits
overall.

A representation that encodes to within the constant factor of the
information-theoretic lower bound (up to lower order terms) and
simultaneously supports queries in optimal time is called a \emph{compact
  data structure}, while if the constant factor is $1$ then it is also a
\emph{succinct data structure} \cite{63533}.

\subsection{Entropy Bounds}
\label{sec:fib-entropy}

A succinct representation very often contains further redundancy in the form
of regularity in the label mapping.  For instance, in the sample trie of
Fig.~\ref{fig:leaf-pushed} there are three leaves with label $2$, but only
one with label $1$ or $3$.  Thus, we could save space by representing label
$2$ on fewer bits, similarly to how Huffman-coding does for strings.  This
correspondence leads to the following notion of \emph{entropy} for
leaf-labeled tries (on the traces of Ferragina \emph{et al.}
\cite{Ferragina:2009:CIL:1613676.1613680}).
\begin{proposition}\label{prop:entropy}
  Let $T$ be a proper, binary, leaf-labeled trie with $n$ leaves on an
  alphabet $\Sigma$, let $p_s$ denote the probability that some symbol $s \in
  \Sigma$ appears as a leaf label, and let $H_0$ denote the Shannon-entropy
  of the probability distribution $p_s, s \in \Sigma$: 
  \begin{equation}\label{eq:Shannon-entropy}
    H_0 = \sum_{s\in\Sigma} p_s \log_2\nicefrac1{p_s} \enspace  .
  \end{equation}
  Then, the zero-order entropy of $T$ is $2n + n H_0$ bits\footnoteE{{\it
      Erratum:} The claim is revised from \cite{sigcomm_2013}, where the
    entropy wrongly appears as $4n + n H_0$. See previous note for an
    explanation.}.
\end{proposition}

Intuitively speaking, the entropy of the tree structure corresponds to the
information-theoretic limit of $2n$ bits as we do not assume any regularity
in this regard.  To this, the leaf-labels add an extra $n H_0$ bits of
entropy.

The entropy of a trie depends mainly on the size of the underlying tree and
the distribution of labels on it.  This transforms to FIBs quite naturally:
the more prefixes go to the same next-hop and the more the FIB resembles ``a
default route with few exceptions'', the smaller the Shannon-entropy of the
next-hop distribution and the tighter the space limit. Accordingly, we shall
define the notions \emph{FIB information-theoretic lower bound} and \emph{FIB
  entropy} as those of the underlying leaf-pushed prefix tree.  Both space
bounds are well-defined as the normalized form is unique. Note, however, that
in contrast to general trees IP FIBs are of bounded height, so the real
bounds should be somewhat smaller. Additionally, for the purposes of this
paper our space bounds involve binary leaf-labeled tries only. We relax this
restriction in \cite{Retvari:2012:CIF:2390231.2390232} using the generic trie
entropy measure of~\cite{Ferragina:2009:CIL:1613676.1613680}.

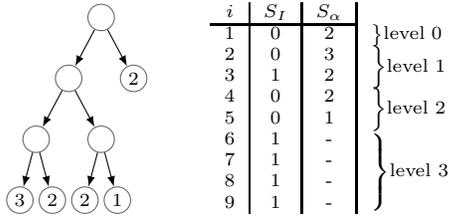
\begin{figure}
  \centering
  \subfigure{%
    \begin{tikzpicture}
      [baseline=(root.base),scale=.27, minimum size=10,inner sep=1pt,
      node/.style={anchor=center,circle,draw=black!60,font=\scriptsize}] {
        \node (0_1) at (0,0) [node] {$3$};%
        \node (0_2) at (1.6,0) [node] {$2$};%
        \node (1_1) at (3.2,0) [node] {$2$};%
        \node (1_2) at (4.8,0) [node] {$1$};%

        \node (i_0) at (.8,3) [node] {}; %
        \path[->,>=latex] (i_0) edge (0_1) edge (0_2);%
        \node (i_1) at (4,3) [node] {}; %
        \path[->,>=latex] (i_1) edge (1_1) edge (1_2);%

        \node (i2_0) at (2.4,6) [node] {}; %
        \path[->,>=latex] (i2_0) edge (i_0) edge (i_1); %
        \node (i2_1) at (5.6,6) [node] {$2$}; %

        \node (root) at (4,9) [node] {}; %
        \path[->,>=latex] (root) edge (i2_0) edge (i2_1);%
      };
    \end{tikzpicture}}%
  \hskip2em%
  \subfigure{%
    \begin{scriptsize}
      \renewcommand{\arraystretch}{1}
      \begin{tabular}[t]{c|c|c|cp{1cm}}
        $i$ & $S_I$ & $S_{\alpha}$\\
        \cline{1-3}
        1 &  0 & 2 & \rdelim\}{1}{3mm}[level 0]\\
        2 &  0 & 3 & \rdelim\}{2}{3mm}[level 1]\\
        3 &  1 & 2 &\\
        4 &  0 & 2 & \rdelim\}{2}{3mm}[level 2]\\
        5 &  0 & 1 &\\
        6 &  1 & - & \rdelim\}{4}{3mm}[level 3]\\
        7 &  1 & - &\\
        8 &  1 & - &\\
        9 &  1 & - &
      \end{tabular}
      \renewcommand{\arraystretch}{1}
    \end{scriptsize}}
  \caption{A leaf-pushed trie and its \XBWL transform.}
  \label{fig:xbwl-trans}
\end{figure}

\section{Attaining Entropy Bounds}
\label{sec:mbw}

Below, we present our first compressed FIB data structure, the
\emph{Burrows-Wheeler transform for binary leaf-labeled tries} (\XBWL).  This
data structure is a stripped down version of \emph{MBW}, the \emph{Multibit
  Burrows-Wheeler} transform from \cite{Retvari:2012:CIF:2390231.2390232},
and the \emph{XBW-l}\xspace transform from \cite{sigcomm_2013}, which in turn
build on the succinct level-indexed binary trees of Jacobson \cite{63533} 
and the XBW transform due to Ferragina \emph{et al.}
\cite{Ferragina:2009:CIL:1613676.1613680}.  In contrast to \XBWL that is
binary only, the original \emph{MBW} and \emph{XBW-l}\xspace transforms
support level-compressed tries as well, at the price of encoding to a
slightly larger representation and missing the information-theoretical
storage size bounds defined above.

The basis for the \XBWL transform is a normalized, binary\footnoteE{{\it
    Erratum:} Originally, \XBWL allowed to encode level-compressed (i.e., not
  necessarily binary) tries as well, which we do not consider here to be able
  to meet the tight information-theoretical bounds.}, proper, leaf-labeled
trie.  Let $T$ be a binary tree on $t$ nodes, let $L$ be the set of leaves
with $n = \abs{L}$, and let $l$ be a mapping $V \mapsto \Sigma$ specifying
for a node $v$ either that $v$ does not have a label associated with it
(i.e., $l(v) = \emptyset$) or the corresponding label $l(v) \in \Sigma$.  If
$T$ is proper, binary, and leaf-labeled, then the following invariants hold:
\begin{description}
\item[P1:] Either $v \in L$, or $v$ has $2$ children.
\item[P2:] $l(v) \neq \emptyset$ $\Leftrightarrow$ $v \in L$.
\item[P3:] $t < 2n$ and so $t=O(n)$.
\end{description}

The main idea in \XBWL is serializing $T$ into a bitstring $S_I$ that encodes
the tree structure and a string $S_{\alpha}$ on the alphabet $\Sigma$
encoding the labels, and then using a sophisticated lossless string
compressor to obtain the storage size bounds\footnoteE{{\it Erratum:} The
  original version also contained a third string $S_{\last}$ that was needed
  to correctly encode level-compressed input, as even with $S_{\last}$
  explicitly stored \XBWL met the (erroneously loose) entropy bound of
  $4n+H_0n$.  Herein, we shave off $S_{\last}$ in order to make up for the
  ``lost'' constant in the storage size bounds.}.  The trick is in making
just the right amount of context available to the string compressors, and
doing this all with guaranteeing optimal lookup on the compressed form.
Correspondingly, the \XBWL transform is defined as the tuple $\xbwl(T) =
(S_{I}, S_{\alpha})$, where
\begin{itemize}
\item $S_I$: a bitstring of size $t$ with zero in position $i$ if the $i$-th
  node of $T$ in level-order is an interior node and $1$ otherwise; and
\item $S_{\alpha}$: a string of size $n$ on the alphabet $\Sigma$ encoding
  the leaf labels.
\end{itemize}

For our sample FIB, the leaf-pushed trie and the corresponding \XBWL
transform are given in Fig.~\ref{fig:xbwl-trans}.

\subsection{Construction and IP lookup}
\label{sec:xbw-constr}

In order to generate the \XBWL transform, one needs to fill up the strings
$S_{I}$ and $S_{\alpha}$, starting from the root and traversing $T$ in a
breadth-first-search order.

\vspace{.3em}\hrule height 1pt\vspace{.2em}%
  \algrenewcommand\algorithmicdo{}
  \begin{small}
    \begin{algorithmic}[1]



      \State{$i \gets 1$; $j \gets 1$}

      \algrenewcommand{\algorithmicfunction}{\textsc{bfs-traverse}}

      \Function{}{node $v$, integer $i$, integer $j$}

      \State{\textbf{if} $v \notin L$ \textbf{then} $S_{I}[i] \gets 0$}

      \State{$\qquad$\textbf{else} $S_{I}[i] \gets 1$;
        $S_{\alpha}[j] \gets l(v)$; $j \gets j + 1$}

      \State{$i \gets i + 1$}


      \EndFunction
    \end{algorithmic}
  \end{small}
\hrule\vspace{.5em}%

The following statement is now obvious.
\begin{lemma}\label{lem:transform}
  Given a proper binary, leaf-labeled trie $T$ on $t$ nodes, $\xbwl(T)$ can
  be built in $O(t)$ time.
\end{lemma}

The transform $\xbwl(T)$ has some appealing properties.  For instance, the
children of some node, if exist, are stored on consecutive indices in $S_{I}$
and $S_{\alpha}$.  In fact, \emph{all} nodes at the same level of $T$ are
mapped to consecutive indices. 

The next step is to actually compress the strings.  This promises easier than
compressing $T$ directly as $\xbwl(T)$, being a sequential string
representation, lacks the intricate structure of tries.  An obvious choice
would be to apply some standard string compressor (like the venerable
\texttt{gzip(1)} tool), but this would not admit queries like ``get all
children of a node'' without first decompressing the transform.  Thus, we
rather use a compressed string self-index \cite{63533,
  Ferragina:2009:CIL:1613676.1613680, Raman:2002:SID:545381.545411,
  Ferragina:2007:CRS:1240233.1240243} to store $\xbwl(T)$, which allows us to
implement efficient navigation immediately on the compressed form.

The way string indexers usually realize navigability is to implement a
certain set of simple primitives in constant $O(1)$ time in-place. Given a
string $S[1,t]$ on alphabet $\Sigma$, a symbol $s \in \Sigma$, and integer $q
\in [1,t]$, these primitives are as follows:
\begin{itemize}
\item $\access(S,q)$: return the symbol at position $q$ in $S$;
\item $\rank_s(S,q)$: return the number of times symbol $s$ occurs in the
  prefix $S[1,q]$; and
\item $\select_s(S,q)$: return the position of the $q$-th occurrence of symbol
  $s$ in $S$.
\end{itemize}

Curiously, these simple primitives admit strikingly complex queries to be
implemented and supported in optimal time.  In particular, the IP lookup
routine on $\xbwl(T)$ takes the following form\footnoteE{{\it Erratum:}
  Pseudo-code updated.}.

\pagebreak
\hrule height 1pt\vspace{.2em}%
  \algrenewcommand{\algorithmicfunction}{$\lookup$}
  \begin{small}
    \begin{algorithmic}[1]
      \Function{}{address $a$}

      \State{$q \gets 0$, $i \gets 1$}

      \While{$q < W$}

      \State{\textbf{if} $\access(S_I,i) = 1$ \textbf{then}}

      \State{$\qquad$\textbf{return} $\access(S_{\alpha}, \rank_1(S_I, i))$}

      \State{$r \gets \rank_0(S_I, i)$}

      \State{$f \gets 2r$}

      \State{$j \gets \getbits(a, q, 1)$}

      \State{$i \gets f + j$; $q \gets q + 1$}

      \EndWhile

      \EndFunction
    \end{algorithmic}
  \end{small}
\hrule\vspace{.5em}%

The code first checks if the actual node, encoded at index $i$ in $\xbwl(T)$,
is a leaf node.  If it is, then $\rank_1(S_I, i)$ tells how many leaves were
found in the course of the BFS-traversal \emph{before} this one and then the
corresponding label is returned from $S_{\alpha}$.  If, on the other hand,
the actual node is an interior node, then $r$ tells how many interior nodes
precede this one and since, as one easily checks, in a level-ordered tree
traversal the children of the $r$-th interior node are encoded from position
$2r$~\cite{63533}, $f$ in fact gets the index of the first child of the
actual node.  Next, we obtain the index $j$ of the child to be visited next
from the address to be looked up, we set the current index to $f+j$ and then
we carry on with the recursion.

\subsection{Memory Size Bounds}
\label{sec:xbw-sucinct}

First, we show that \XBWL is a succinct FIB representation, in that it
supports lookup in optimal $O(W)$ time and encodes to information-theoretic
lower bound.
\begin{lemma}\label{lem:inf_theo_min}
  Given a proper, binary, leaf-labeled trie $T$ with $n$ leaves on an
  alphabet of size $\delta$, $\xbwl(T)$ can be stored on $2n + n \lg\delta$
  bits so that $\lookup$ on $\xbwl(T)$ terminates in $O(W)$
  time\footnoteE{{\it Erratum:} Claim updated according to the correct
    information-theoretical lower-bound.}.
\end{lemma}
\begin{proof}
  One can encode $S_I$ on at most $t \approx 2n$ bits using the \texttt{RRR}
  succinct bitstring index \cite{Raman:2002:SID:545381.545411}, which
  supports $\select$ and $\rank$ in $O(1)$.  In addition, even the trivial
  encoding of $S_{\alpha}$ needs only another $n \lg\delta$ bits and provides
  $\access$ in $O(1)$.  So every iteration of $\lookup$ takes constant time,
  which gives the result.
\end{proof}

Next, we show that \XBWL can take advantage of regularity in leaf labels (if
any) and encode \emph{below} the information-theoretic bound to zero-order
entropy.
\begin{lemma}\label{lem:H_0}
  Let $T$ be a proper, binary, leaf-labeled trie with $n$ leaves on an
  alphabet of size $O(\polylog n)$, and let $H_0$ denote the Shannon-entropy
  of the leaf-label distribution.  Then, $\xbwl(T)$ can be stored on $2n +
  nH_0 + o(n)$ bits so that $\lookup$ on $\xbwl(T)$ terminates in $O(W)$
  time\footnoteE{{\it Erratum:} Claim updated according to the correct
    entropy bound.}.
\end{lemma}
\begin{proof}
  $S_I$ can be encoded as above, and $S_{\alpha}$ can be stored on $nH_0 +
  o(n)$ bits using generalized wavelet trees so that $\access$ is $O(1)$,
  under the assumption that the alphabet size is $O(\polylog n)$
  \cite{Ferragina:2007:CRS:1240233.1240243}.
\end{proof}

Interestingly, the above \emph{zero-order} entropy bounds can be easily
upgraded to \emph{higher-order} entropy.  A fundamental premise in data
compression is that elements in a data set often depend on their neighbors,
and the larger the context the better the prediction of the element from its
context and the more efficient the compressor.  Higher-order string
compressors can use the Burrows-Wheeler transform to exploit this contextual
dependency, a reversible permutation that places symbols with similar context
close to each other. This argumentation readily generalizes to leaf-labeled
tries; simply, the \emph{context of a node corresponds to its level in the
  tree} and because \XBWL organizes nodes at the same level (i.e., of similar
context) next to each other, it realizes the same effect for tries as the
Burrows-Wheeler transform for strings (hence the name).  Deciding whether or
not such contextual dependency is present in real IP FIBs is well beyond the
scope of this paper.  Here, we only note that if it is, then \XBWL can take
advantage of this and compress an IP FIB to higher-order entropy using the
techniques in \cite{Ferragina:2009:CIL:1613676.1613680,
  Retvari:2012:CIF:2390231.2390232}.

\begin{figure*}
  \centering
  \subfigure[][\mbox{\hspace{30pt}}]{%
    \begin{tikzpicture}%
      [baseline=(root.base),scale=.22, minimum size=9,inner sep=1pt,
      node/.style={anchor=center,circle,draw=black!60,font=\scriptsize}] {
        \useasboundingbox (0,.8) rectangle (14,12);
        \node (_0_0_0) at (9,0) [node] {$1$};%

        \node (0_0_1) at (2,3) [node] {$2$};%

        \node (0_2_0) at (6,3) [node] {$1$};%
        \node (0_3_0) at (10,3) [node] {};%
        \path[->,>=latex] (0_3_0)%
			edge%
                            node[above left,minimum size=1ex,scale=.7]{$0$}%
					(_0_0_0)%
			;%

		\node (0_3_1) at (12.4,3) [node] {$3$};%

        \node (1_0_0) at (1,6) [node] {}; %
        \path[->,>=latex] (1_0_0)%
				edge %
					node[above right,minimum size=1ex,scale=.7]{$1$}%
						(0_0_1);%
		\node (1_0_1) at (4,6) [node] {$3$}; %

        \node (1_1_0) at (7,6) [node] {}; %
        \path[->,>=latex] (1_1_0)%
			edge%
				node[above left,minimum size=1ex,scale=.7]{$0$}%
					(0_2_0)%
			;%

        \node (1_1_1) at (11,6) [node] {}; %
        \path[->,>=latex] (1_1_1)%
			edge%
                                node[above left,minimum size=1ex,scale=.7]{$0$}%
					(0_3_0)%
			edge%
				node[above right,minimum size=1ex,scale=.7]{$1$}%
					(0_3_1);%

        \node (2_0_0) at (2.5,9) [node] {}; %
        \path[->,>=latex] (2_0_0)%
			edge%
				node[above left,minimum size=1ex,scale=.7]{$0$}%
					(1_0_0)%
			edge%
				node[above right,minimum size=1ex,scale=.7]{$1$}%
					(1_0_1);%

        \node (2_0_1) at (9,9) [node] {$2$}; %
        \path[->,>=latex] (2_0_1)%
			edge%
				node[above left,minimum size=1ex,scale=.7]{$0$}%
					(1_1_0)%
			edge%
				node[above right,minimum size=1ex,scale=.7]{$1$}%
					(1_1_1);%

        \node (root) at (5.75,12) [node] {$1$}; %
        \path[->,>=latex] (root)%
			edge%
				node[above left,minimum size=1ex,scale=.7]{$0$}%
					(2_0_0)%
			edge%
				node[above right,minimum size=1ex,scale=.7]{$1$}%
					(2_0_1);%
      };
    \end{tikzpicture}\label{fig:tf-pretix-trie}}%
  \hskip.5em%
  \subfigure[][\mbox{\hspace{30pt}}]{%
    \begin{tikzpicture}%
      [baseline=(root.base),scale=.22, minimum size=9,inner sep=1pt,
      node/.style={anchor=center,circle,draw=black!60,font=\scriptsize}] {
        \useasboundingbox (0,.8) rectangle (14,12);
        \node (_0_0_0) at (9,0) [node] {$1$};%
        \node (_0_0_1) at (11,0) [node] {$2$};%

        \node (0_0_0) at (0,3) [node] {$1$};%
        \node (0_0_1) at (2,3) [node] {$2$};%

        \node (0_2_0) at (6,3) [node] {$1$};%
        \node (0_2_1) at (8,3) [node] {$2$};%
        \node (0_3_0) at (10,3) [node] {};%
        \path[->,>=latex] (0_3_0)%
			edge%
					(_0_0_0)%
			edge%
					(_0_0_1);%

		\node (0_3_1) at (12.4,3) [node] {$3$};%

        \node (1_0_0) at (1,6) [node] {}; %
        \path[->,>=latex] (1_0_0)%
				edge %
				 		(0_0_0)
				edge %
						(0_0_1);%
		\node (1_0_1) at (4,6) [node] {$3$}; %

        \node (1_1_0) at (7,6) [node] {}; %
        \path[->,>=latex] (1_1_0)%
			edge%
					(0_2_0)%
			edge%
					(0_2_1);%

        \node (1_1_1) at (11,6) [node] {}; %
        \path[->,>=latex] (1_1_1)%
			edge%
					(0_3_0)%
			edge%
					(0_3_1);%

        \node (2_0_0) at (2.5,9) [node] {}; %
        \path[->,>=latex] (2_0_0)%
			edge%
					(1_0_0)%
			edge%
					(1_0_1);%

        \node (2_0_1) at (9,9) [node] {}; %
        \path[->,>=latex] (2_0_1)%
			edge%
					(1_1_0)%
			edge%
					(1_1_1);%

        \node (root) at (5.75,12) [node] {\textcolor{white}{$a$}}; %
        \path[->,>=latex] (root)%
			edge%
					(2_0_0)%
			edge%
					(2_0_1);%
	\draw[dash pattern=on 0.30mm off 0.50mm](3,12) -- (8,12);%
      };
    \end{tikzpicture}\label{fig:tf-lp-0}}%
 \hskip.5em%
  \subfigure[]{%
    \begin{tikzpicture}%
      [baseline=(root.base),scale=.22, minimum size=9,inner sep=1pt,
      node/.style={anchor=center,circle,draw=black!60,font=\scriptsize}] {
        \useasboundingbox (0,.8) rectangle (10,12);
        \node[draw=none] (_0_0_0) at (9,0) {}; 
        \node (0_0_0) at (0,3) [node] {$1$};%
        \node (0_0_1) at (2,3) [node] {$2$};%



        \node (1_0_0) at (1,6) [node] {}; %
        \path[->,>=latex] (1_0_0)%
				edge %
				 		(0_0_0)
				edge %
						(0_0_1);%
		\node (1_0_1) at (4,6) [node] {$3$}; %



        \node (2_0_0) at (2.5,9) [node] {}; %
        \path[->,>=latex] (2_0_0)%
			edge%
					(1_0_0)%
			edge%
					(1_0_1);%

        \node (2_0_1) at (9,9) [node] {}; %
        \path[->,>=latex] (2_0_1)%
			edge%
				node[above left,minimum size=1ex,scale=.7]{$0$}%
					(1_0_0)%
			edge%
				node[above,minimum size=1ex,scale=.7]{$1$}%
					(2_0_0);%

        \node (root) at (5.75,12) [node] {\textcolor{white}{$a$}}; %
        \path[->,>=latex] (root)%
			edge%
					(2_0_0)%
			edge%
					(2_0_1);%
	\draw[dash pattern=on 0.30mm off 0.50mm](3,12) -- (8,12);%
      };
    \end{tikzpicture}\label{fig:tf-dag-0}}%
  \subfigure[][\mbox{\hspace{30pt}}]{%
    \begin{tikzpicture}%
      [baseline=(root.base),scale=.22, minimum size=9,inner sep=1pt,
      node/.style={anchor=center,circle,draw=black!60,font=\scriptsize}] {
        \useasboundingbox (0,.8) rectangle (14,12);
        \node (_0_0_0) at (9,0) [node] {$1$};%
        \node (_0_0_1) at (11,0) [node] {$2$};%

        \node (0_0_0) at (0,3) [node] {$\perp$};%
        \node (0_0_1) at (2,3) [node] {$2$};%

        \node (0_2_0) at (6,3) [node] {$1$};%
        \node (0_2_1) at (8,3) [node] {$2$};%
        \node (0_3_0) at (10,3) [node] {};%
        \path[->,>=latex] (0_3_0)%
			edge%
					(_0_0_0)%
			edge%
					(_0_0_1);%

		\node (0_3_1) at (12.4,3) [node] {$3$};%

        \node (1_0_0) at (1,6) [node] {}; %
        \path[->,>=latex] (1_0_0)%
				edge %
				 		(0_0_0)
				edge %
						(0_0_1);%
		\node (1_0_1) at (4,6) [node] {$3$}; %

        \node (1_1_0) at (7,6) [node] {}; %
        \path[->,>=latex] (1_1_0)%
			edge%
					(0_2_0)%
			edge%
					(0_2_1);%

        \node (1_1_1) at (11,6) [node] {}; %
        \path[->,>=latex] (1_1_1)%
			edge%
					(0_3_0)%
			edge%
					(0_3_1);%

        \node (2_0_0) at (2.5,9) [node] {}; %
        \path[->,>=latex] (2_0_0)%
			edge%
					(1_0_0)%
			edge%
					(1_0_1);%

        \node (2_0_1) at (9,9) [node] {}; %
        \path[->,>=latex] (2_0_1)%
			edge%
					(1_1_0)%
			edge%
					(1_1_1);%

        \node (root) at (5.75,12) [node] {$1$}; %
        \path[->,>=latex] (root)%
			edge%
					(2_0_0)%
			edge%
					(2_0_1);%
        \draw[dash pattern=on 0.30mm off 0.50mm](0.5,9) --%
                  (11.5,9);
      };
    \end{tikzpicture}\label{fig:tf-lp-1}}%
  \hskip.5em%
  \subfigure[]{%
    \begin{tikzpicture}%
      [baseline=(root.base),scale=.22, minimum size=9,inner sep=1pt,
      node/.style={anchor=center,circle,draw=black!60,font=\scriptsize}] {
        \useasboundingbox (0,.8) rectangle (12,12);
        \node[draw=none] (_0_0_0) at (9,0) {}; 
        \node (0_0_0) at (0,3) [node] {};

        \node (0_2_0) at (6,3) [node] {$1$};%
        \node (0_2_1) at (8,3) [node] {$2$};%


        \node (1_0_0) at (1,6) [node] {}; %
        \path[->,>=latex] (1_0_0)%
				edge %
					node[above left,minimum size=1ex,scale=.7]{$0$}%
				 		(0_0_0)
				edge[out=-30,in=135] %
					node[below,minimum size=1ex,scale=.7]{$1$}%
						(0_2_1);%
		\node (1_0_1) at (4,6) [node] {$3$}; %

        \node (1_1_0) at (7,6) [node] {}; %
        \path[->,>=latex] (1_1_0)%
			edge%
					(0_2_0)%
			edge%
					(0_2_1);%

        \node (1_1_1) at (11,6) [node] {}; %
        \path[->,>=latex] (1_1_1)%
			edge%
				node[below right,minimum size=1ex,scale=.7]{$0$}%
					(1_1_0)%
			edge[out=145,in=35]%
				node[above left,minimum size=1ex,scale=.7]{$1$}%
					(1_0_1);%

        \node (2_0_0) at (2.5,9) [node] {}; %
        \path[->,>=latex] (2_0_0)%
			edge%
					(1_0_0)%
			edge%
					(1_0_1);%

        \node (2_0_1) at (9,9) [node] {}; %
        \path[->,>=latex] (2_0_1)%
			edge%
					(1_1_0)%
			edge%
					(1_1_1);%

        \node (root) at (5.75,12) [node] {$1$}; %
        \path[->,>=latex] (root)%
			edge%
					(2_0_0)%
			edge%
					(2_0_1);%
		\draw[dash pattern=on 0.30mm off 0.50mm](0.5,9) --%
                        (11.5,9);
      };
    \end{tikzpicture}\label{fig:tf-dag-1}}%
  \hskip.5em%
  \subfigure[]{%
    \begin{tikzpicture}%
      [baseline=(root.base),scale=.22, minimum size=9,inner sep=1pt,
      node/.style={anchor=center,circle,draw=black!60,font=\scriptsize}] {
        \useasboundingbox (0,.8) rectangle (12,12);
        \node[draw=none] (_0_0_0) at (9,0) {}; 
        \node (0_0_1) at (2,3) [node] {$2$};%

        \node (0_2_0) at (6,3) [node] {$1$};%
        \node (0_2_1) at (8,3) [node] {};


        \node (1_0_0) at (1,6) [node] {}; %
        \node (1_1_0) at (7,6) [node] {}; %
        \path[->,>=latex] (1_1_0)%
			edge%
					(0_2_0)%
			edge%
					(0_2_1);%

        \node (1_1_1) at (11,6) [node] {}; %
        \node (2_0_0) at (2.5,9) [node] {}; %
        \path[->,>=latex] (2_0_0)%
			edge%
					(1_0_0)%
			edge%
					(1_0_1);%

        \node (2_0_1) at (9,9) [node] {$2$}; %
        \path[->,>=latex] (2_0_1)%
			edge%
					(1_1_0)%
			edge%
					(1_1_1);%

        \node (root) at (5.75,12) [node] {$1$}; %
        \path[->,>=latex] (root)%
			edge%
					(2_0_0)%
			edge%
					(2_0_1);%

	\draw[dash pattern=on 0.30mm off 0.50mm] (0,6) --%
		(12.5,6);
        \path[->,>=latex] (1_0_0)%
				edge[out=-35,in=135] %
					node[below left,minimum size=1ex,scale=.7]{$0$}%
				 		(0_2_1)
				edge %
					node[above right,minimum size=1ex,scale=.7]{$1$}%
						(0_0_1);%
		\node (1_0_1) at (4,6) [node] {$3$}; %

        \path[->,>=latex] (1_1_1)%
			edge%
				node[below,minimum size=1ex,scale=.7]{$0$}%
					(1_1_0)%
			edge[out=145,in=35]%
				node[above left,minimum size=1ex,scale=.7]{$1$}%
					(1_0_1);%

      };
    \end{tikzpicture}\label{fig:tf-dag-2}}%
  \caption{A binary trie for a sample FIB (a); the same trie when
    leaf-pushing is applied from level $\lambda=0$ (b); the prefix DAG for
    leaf-push barrier $\lambda=0$ (c); the trie (d) and the prefix DAG (e)
    for $\lambda=1$; and the prefix DAG for $\lambda=2$ (f).  Dashed lines
    indicate the leaf-push barrier $\lambda$ and the invalid label $\perp$
    was removed from the corresponding leaf nodes in the prefix DAGs.}
\label{fig:trie-fold}
\end{figure*}
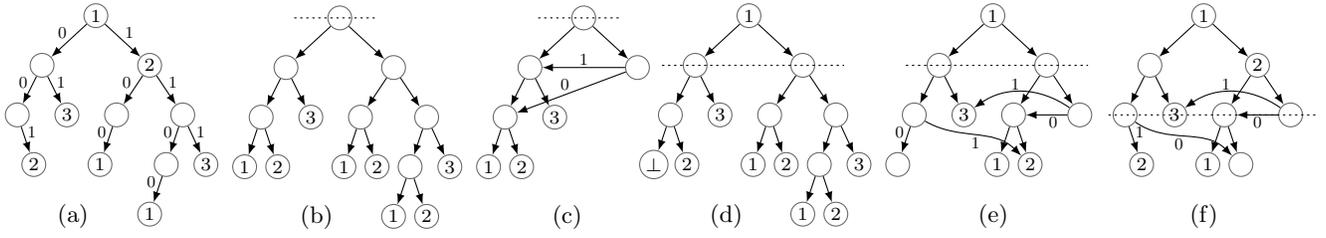

In summary, the \XBWL transform can be built fast, supports lookup in
asymptotically optimal time, and compresses to within entropy
bounds. Updates, however, may be expensive. Even the underlying leaf-pushed
trie takes $O(n)$ steps in the worst-case to update, after which we could
either rebuild the string indexes from scratch (again in $O(n)$) or use a
dynamic compressed index that supports updates to the compressed form
efficiently.  For instance, \cite{Makinen:2008:DES:1367064.1367072}
implements insertion and deletion in roughly $O(\log n)$ time, at the cost of
slower $\rank$ and $\select$.  The other shortcoming of \XBWL is that, even
if it supports lookups in \emph{theoretically} optimal time, it is just too
slow for line speed IP lookup (see Section~\ref{sec:num-lookup}).  In the
next section, therefore, we discuss a somewhat more practical FIB compression
scheme.

\section{Practical FIB Compression}
\label{sec:trie-fold}

The string indexes that underly \XBWL are \emph{pointerless}, encoding all
information in compact bitmaps.  This helps squeezing \XBWL into higher-order
entropy bounds but also causes that we need to perform multiple $\rank$ and
$\select$ operations just to, say, enter a child of a node.  And even though
these primitives run in $O(1)$ the constants still add up, building up delays
too huge for line speed IP lookup.  In contrast, a traditional \emph{pointer
  machine}, like a prefix tree, can follow a child pointer in just a single
indirection with only one random memory access overhead.  The next compressed
FIB data structure we introduce is, therefore, pointer-based.

The idea is to essentially re-invent the classic prefix tree, borrowing the
basic mechanisms from the Lempel-Ziv (\texttt{LZ78}) string compression
scheme \cite{Cover:1991:EIT:129837}.  \texttt{LZ78} attains entropy by
parsing the text into unique sub-strings, yielding a form that contains no
repetitions.  Tree threading is a generalization of this technique to
\emph{unlabeled trees}, converting the tree into a Directed Acyclic Graph
(DAG) by merging isomorphic sub-trees \cite{KATAJAINEN:1990p202, 4694879,
  Ioannidis:2005:LCD:1114718.1648670, DBLP:conf/icnp/SongKHL09}. In this
paper, we apply this idea to \emph{labeled trees}, merging sub-tries taking
into account both the underlying tree structure \emph{and} the labels
\cite{Bryant:1992:SBM:136035.136043, Cocke:1970:GCS:390013.808480}.  If the
trie is highly regular then this will eliminate all recurrent sub-structures,
producing a representation that contains no repetitions and hence, so the
argument goes, admits entropy bounds like \texttt{LZ78}. 

The below equivalence relation serves as the basis of our trie-merging
technique.
\begin{definition}\label{def:trie-equiv}
  Two leaf-labeled tries are identical if each of their sub-tries are
  pairwise identical, and two leaves are identical if they hold the same
  label.
\end{definition}

We call the algorithmic manifestation of this recursion the
\emph{trie-folding algorithm} and the resultant representation a \emph{prefix
  DAG}.  Fig.~\ref{fig:tf-pretix-trie} depicts a sample prefix tree,
Fig.~\ref{fig:tf-lp-0} shows the corresponding leaf-pushed trie, and
Fig.~\ref{fig:tf-dag-0} gives the prefix DAG.  For instance, the sub-tries
that belong to the prefix $0/1$ and $11/2$ are equivalent in the leaf-pushed
trie, and thusly merged into a single sub-trie that is now available in the
prefix DAG along two paths from the root.  This way, the prefix DAG is
significantly smaller than the original prefix tree as it contains only half
the nodes.

For the trie-folding algorithm it is essential that the underlying trie be
normalized; for instance, in our example it is easy to determine from the
leaf-pushed trie that the two sub-tries under the prefixes $00/2$ and $10/2$
are identical, but this is impossible to know from the original prefix tree.
Thus, leaf-pushing is essential to realize good compression but, at the same
time, makes updates prohibitive \cite{DBLP:conf/icnp/SongKHL09}.

To avoid this problem, we apply a simple optimization.  We separate the trie
into two parts; ``above'' a certain level $\lambda$, called the
\emph{leaf-push barrier}, where sub-tries are huge and so common sub-tries
are rare, we store the FIB as a standard binary prefix tree in which update
is fast; and ``below'' $\lambda$, where common sub-tries most probably turn
up, we apply leaf-pushing to obtain good compression. Then, by a cautious
setting of the leaf-push barrier we simultaneously realize fast updates and
entropy-bounded storage size.

The prefix DAGs for $\lambda=1$ and $\lambda=2$ are depicted in
Fig.~\ref{fig:tf-dag-1} and \ref{fig:tf-dag-2}.  The size is somewhat larger,
but updating, say, the default route now only needs setting the root label
without having to cycle through each leaf.

\subsection{Construction and IP lookup}
\label{sec:trie-folding:constuction}

We are given a binary trie $T$ (not necessarily proper and leaf-pushed) of
depth $W$, labeled from an alphabet $\Sigma$ of size $\delta$.  Let $V_T$
($\abs{V_T}=t$) be the set of nodes and $L_T$ ($\abs{L_T}=n$) be the set of
leaves.  Without loss of generality, we assume that $T$ does not contain
explicit blackhole routes. Then, the trie-folding algorithm transforms $T$
into a prefix DAG $D(T)$, on nodes $V_D$ and leaves $L_D$, with respect to
the leaf-push barrier $\lambda \in [0,W]$.

The algorithm is a simple variant of trie threading
\cite{KATAJAINEN:1990p202}: assign a unique id to each sub-trie that occurs
below level $\lambda$ and merge two tries if their ids are equal (as of
Definition~\ref{def:trie-equiv}). The algorithm actually works on a copy of
$T$ and always keeps an intact instance of $T$ available.  This instance,
called the control FIB, can exist in the DRAM of the line card's control CPU,
as it is only consulted to manage the FIB.  The prefix DAG itself is
constructed in fast memory.  We also need two ancillary data structures, the
leaf table and the sub-trie index, which can also live in DRAM.

The \emph{leaf table} will be used to coalesce leaves with identical labels
into a common leaf node.  Accordingly, for each $s \in \Sigma$ the leaf table
$\lp(s)$ stores a leaf node 
(no matter which one) with that label.  Furthermore, the \emph{sub-trie
  index} $\stind$ will be used to identify and share common sub-tries.
$\stind$ is in fact a reference counted associative array, addressed with
pairs of ids $(i, j) \in \mathbb{N} \times \mathbb{N}$ as keys and storing
for each key a node whose children are exactly the sub-tries identified by
$i$ and $j$. $\stind$ supports the following primitives:
\begin{itemize}
\item $\Put(i, j, v)$: if a node with key $(i,j)$ exists in $\stind$ then
  increase its reference count and return it, otherwise generate a new id in
  $v.\id$, store $v$ at key $(i,j)$ with reference count $1$, and return $v$;
  and
\item $\Get(i, j)$: dereference the entry with key $(i,j)$ and delete it if
  the reference count drops to zero.
\end{itemize}

In our code we used a hash to implement $\stind$, which supports the above
primitives in amortized $O(1)$ time.  

Now, suppose we are given a node $v$ to be subjected to trie-folding and a
leaf-push barrier $\lambda$.  First, for each descendant $u$ of $v$ at depth
$\lambda$ we normalize the sub-trie rooted at $u$ using label $l(u)$ as a
``default route'', and then we call the compress routine to actually merge
identical leaves and sub-tries below $u$, starting from the bottom and
working upwards until we reach $u$. Consider the below pseudo-code for the
main $\triefold$ routine.

\vspace{.3em}\hrule height 1pt\vspace{.2em}%
  \algrenewcommand\algorithmicdo{}%
  \begin{small}
    \begin{algorithmic}[1]
      \algrenewcommand{\algorithmicfunction}{$\triefold$}

      \Function{}{node $v$, integer $\lambda$}

      \For{\textbf{each} $\lambda$-level child $u$ of $v$ \textbf{do}}

      \State{\textbf{if} $l(u) = \emptyset$} 

      \State{$\qquad$\textbf{then} $\leafpush(u, \perp)$ \textbf{else}
        $\leafpush(u, l(u))$}

      \State{\textsc{postorder-traverse-at-$u$}($\compress$)}

      \EndFor

      \State{$l(\lp(\perp)) \gets \emptyset$}

      \EndFunction

      \algrenewcommand{\algorithmicfunction}{$\compress$}

      \Function{}{node $w$}

      \State{\textbf{if} $w \in L_D$ \textbf{then} $w.\id=l(w)$; $u \gets \lp(w)$}

      \State{$\qquad$\textbf{else} $u = \Put(w.\Left.\id, w.\Right.\id, w)$}

      \State{\textbf{if} $u \neq w$ \textbf{then} re-pointer the parent of $w$ to $u$; $\Delete(w)$}

      \EndFunction
    \end{algorithmic}
  \end{small}
\hrule\vspace{.5em}%

Here, $w.\Left$ is the left child and $w.\Right$ is the right child for $w$,
and $w.\id$ is the id of $w$.  As $\triefold$ visits each node at most twice
and $\compress$ runs in $O(1)$ if $\Put$ is $O(1)$, we arrive to the
following conclusion.
\begin{lemma}
  Given a binary trie $T$ on $t$ nodes, $D(T)$ can be constructed in $O(t)$
  time.
\end{lemma}

Lookup on a prefix DAG goes exactly the same way as on a conventional prefix
tree: follow the path traced out by the successive bits of the lookup key and
return the last label
found. 
We need to take care of a subtlety in handling invalid labels, though.
Namely, in our sample FIB of Fig.~\ref{fig:tf-pretix-trie}, the address $000$
is associated with label $1$ (the default route), which in the prefix DAG for
$\lambda=1$ (Fig.~\ref{fig:tf-dag-1}, derived from the trie on
Fig.~\ref{fig:tf-lp-1}) would become $\perp$ if we were to let leaf nodes'
$\perp$ labels override labels inherited from levels above $\lambda$.  This
problem is easy to overcome, though, by removing the label from the leaf
$\lp(\perp)$.  By our assumption the FIB contains no explicit blackhole
routes, thus every traversal yielding the empty label on $T$ terminates in
$\lp(\perp)$ on $D(T)$ and, by the above modification, also gives an empty
label.

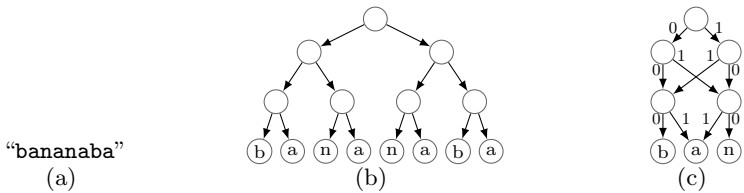
\begin{figure*}
  \centering
  \subfigure[]{%
    ``\texttt{bananaba}''\label{fig:tf-string}}%
  \hskip5em%
  \subfigure[]{%
    \begin{tikzpicture}%
      [baseline=(0_0_0.base),scale=.22, minimum size=9,inner sep=1pt,
      node/.style={anchor=center,circle,draw=black!60,font=\scriptsize}] {
        \node (0_0_0) at (0,0) [node] {b};%
        \node (0_0_1) at (2,0) [node] {a};%
        \node (0_1_0) at (4,0) [node] {n};%
        \node (0_1_1) at (6,0) [node] {a};%
        \node (0_2_0) at (8,0) [node] {n};%
        \node (0_2_1) at (10,0) [node] {a};%
        \node (0_3_0) at (12,0) [node] {b};%
		\node (0_3_1) at (14,0) [node] {a};%

        \node (1_0_0) at (1,3) [node] {}; %
        \path[->,>=latex] (1_0_0)%
				edge %
				 		(0_0_0)
				edge %
						(0_0_1);%
		\node (1_0_1) at (5,3) [node] {}; %
        \path[->,>=latex] (1_0_1)%
				edge %
				 		(0_1_0)
				edge %
						(0_1_1);%

        \node (1_1_0) at (9,3) [node] {}; %
        \path[->,>=latex] (1_1_0)%
			edge%
					(0_2_0)%
			edge%
					(0_2_1);%

        \node (1_1_1) at (13,3) [node] {}; %
        \path[->,>=latex] (1_1_1)%
			edge%
					(0_3_0)%
			edge%
					(0_3_1);%

        \node (2_0_0) at (3,6) [node] {}; %
        \path[->,>=latex] (2_0_0)%
			edge%
					(1_0_0)%
			edge%
					(1_0_1);%

        \node (2_0_1) at (11,6) [node] {}; %
        \path[->,>=latex] (2_0_1)%
			edge%
					(1_1_0)%
			edge%
					(1_1_1);%

        \node (root) at (7,8) [node] {}; %
        \path[->,>=latex] (root)%
			edge%
					(2_0_0)%
			edge%
					(2_0_1);%
      };
    \end{tikzpicture}\label{fig:tf-string-tree}}%
  \hskip6em%
  \subfigure[]{%
    \begin{tikzpicture}%
      [baseline=(0_0_0.base),scale=.22, minimum size=9,inner sep=1pt,
      node/.style={anchor=center,circle,draw=black!60,font=\scriptsize}] {
        \node (0_0_0) at (1,0) [node] {b};%
        \node (0_0_1) at (3,0) [node] {a};%
        \node (0_1_0) at (5,0) [node] {n};%

        \node (1_0_0) at (1,3) [node] {}; %
        \path[->,>=latex] (1_0_0)%
				edge %
					node[above left,minimum size=1ex,scale=.7]{$0$}%
				 		(0_0_0)
				edge %
					node[above right,minimum size=1ex,scale=.7]{$1$}%
						(0_0_1);%
		\node (1_0_1) at (5,3) [node] {}; %
        \path[->,>=latex] (1_0_1)%
				edge %
					node[above right,minimum size=1ex,scale=.7]{$0$}%
				 		(0_1_0)
				edge %
					node[above left,minimum size=1ex,scale=.7]{$1$}%
						(0_0_1);%

        \node (2_0_0) at (1,6) [node] {}; %
        \path[->,>=latex] (2_0_0)%
			edge%
				node[above left,minimum size=1ex,scale=.7]{$0$}%
					(1_0_0)%
			edge%
				node[above=8pt,left=3pt,minimum size=1ex,scale=.7]{$1$}%
					(1_0_1);%

        \node (2_0_1) at (5,6) [node] {}; %
        \path[->,>=latex] (2_0_1)%
			edge%
				node[above=8pt,right=3pt,minimum size=1ex,scale=.7]{$1$}%
					(1_0_0)%
			edge%
				node[above right,minimum size=1ex,scale=.7]{$0$}%
					(1_0_1);%

        \node (root) at (3,8) [node] {}; %
        \path[->,>=latex] (root)%
			edge%
				node[above left,minimum size=1ex,scale=.7]{$0$}%
					(2_0_0)%
			edge%
				node[above right,minimum size=1ex,scale=.7]{$1$}%
					(2_0_1);%
      };
    \end{tikzpicture}\label{fig:tf-string-dag}}%
  \caption{Trie-folding as string compression: a string (a), the complete
    binary trie (b), and the compressed DAG (c).  The third character of the
    string can be accessed by looking up the key $3-1 = 010_2$.}
\label{fig:tf-string-compress}
\end{figure*}

That being said, the last line of the $\triefold$ algorithm renders standard
trie lookup correct on prefix DAGs. Since this is precisely the lookup
algorithm implemented in many IP routers on the ASIC \cite{ezchip}, we
conclude that prefix DAGs can serve as compressed drop-in replacements for
trie-based FIBs in many router products (similarly to e.g.,
\cite{Uzmi:2011:SPN:2079296.2079325}).

The following statement is now obvious.
\begin{lemma}
  The $\lookup$ operation on $D(T)$ terminates in $O(W)$ time.
\end{lemma}

In this regard, trie-folding can be seen as a generalization of conventional
FIB implementations: for $\lambda=32$ we get good old prefix trees, and for
smaller settings of $\lambda$ we obtain \emph{increasingly smaller FIBs with
  exactly zero cost on lookup efficiency}.  Correspondingly, there is no
memory size vs. lookup complexity ``space-time'' trade-off in trie-folding.

\subsection{Memory Size Bounds}
\label{sec:trie-folding:storage}

The tries that underlie \XBWL are proper and leaf-labeled, and the nice
structure makes it easy to reason about the size thereof.  Unfortunately, the
tries that prefix DAGs derive from are of arbitrary shape and so it is
difficult to infer space bounds in the same generic sense.  We chose a
different approach, therefore, in that we \emph{view trie-folding as a
  generic string compression method and we compare the size of the prefix DAG
  to that of an input string} given to the algorithm.  The space bounds
obtained this way transform to prefix trees naturally, as trie entropy itself
is also defined in terms of string entropy (recall Proposition
\ref{prop:entropy}).

Instead of being given a trie, therefore, we are now given a string $S$ of
length $n$ on an alphabet of size $\delta$ and zero-order entropy $H_0$.
Supposing that $n$ equals some power of $2$ (which we do for the moment),
say, $n=2^W$, we can think as if the symbols in $S$ were written to the
leaves of a complete binary tree of depth $W$ as labels.  Then, trie-folding
will convert this tree into a prefix DAG $D(S)$, and we are curious as to how
the size of $D(S)$ relates to the information-theoretic limit for storing
$S$, that is, $n \lg \delta$, and the zero order entropy $n H_0$ (see
Fig.~\ref{fig:tf-string-compress}).  Note that every FIB has such a
``complete binary trie'' representation, and vice versa.

The memory model for storing the prefix DAG is as follows. Above the
leaf-push barrier $\lambda$ we use the usual trick that the children of a
node are arranged on consecutive memory locations~\cite{772439}, and so each
node holds a single node pointer of size to be determined later, plus a label
index of $\lg \delta$ bits. At and below level $\lambda$ nodes hold two
pointers but no label, plus we also need an additional $\delta \lg \delta$
bits to store the coalesced leaves.

Now, we are in a position to state the first space bound.  In particular, we
show that $D(S)$ attains information-theoretic lower bound up to some small
constant factor, and so it is a \emph{compact data structure}.  Our result
improves the constant term in the bound available in \cite{sigcomm_2013} from
$5$ to $4$.
\begin{theorem}\label{thm:proof_compact_data}
  Let $S$ be a string of length $n=2^W$ on an alphabet of size $\delta$ and
  set the leaf-push barrier as

  \begin{equation}
    \label{eq:leaf-push-barrier-delta}
    \lambda = \Big \lfloor \frac1{\ln 2} \lambertW\left( n \ln \delta
    \right) \Big \rfloor \enspace ,
  \end{equation}
  where $\lambertW()$ denotes the Lambert $\lambertW$-function.  Then, $D(S)$
  can be encoded on at most $4 \lg(\delta)n + o(n)$ bits.
\end{theorem}

Note that the Lambert $\lambertW$-function (or product logarithm)
$\lambertW(z)$ is defined as $z = \lambertW(z) e^{\lambertW(z)}$.  The
detailed proof, based on a counting argument, is deferred to the Appendix.

Next, we show that trie-folding compresses to within a constant factor of the
zero-order entropy bound, subject to some reasonable assumptions on the
alphabet.  Furthermore, the constant term is improved from $7$ as of
\cite{sigcomm_2013} to $6$.
\begin{theorem}\label{thm:proof_entropy}
  Let $S$ be a string of length $n$ and zero-order entropy $H_0$, and set the
  leaf-push barrier as
  \begin{equation}
    \label{eq:leaf-push-barrier-H0}
    \lambda = \Big \lfloor  \frac1{\ln 2} \lambertW(n H_0 \ln 2)  \Big
    \rfloor \enspace .
  \end{equation}
  Then, the expected size of $D(S)$ is at most $(6 + 2\lg \frac1{H_0} +
  2\lg\lg \delta )H_0n + o(n)$ bits.
\end{theorem}

Again, refer to the Appendix for the proof.

It turns out that the compression ratio depends on the specifics of the
alphabet. For reasonable $\delta$, say, $\delta=O(1)$ or $\delta=O(\polylog
n)$, the error $\lg\lg\delta$ is very small and the bound gradually improves
as $H_0$ increases, to the point that at maximum entropy $H_0 = \lg \delta$
we get precisely $6H_0 n$. For extremely small entropy, however, the error
$2\lg \frac1{H_0}$ can become dominant as the overhead of the DAG outweighs
the size of the very string in such cases.

\subsection{Update}
\label{sec:trie-folding:dynamic}

What remained to be done is to set the leaf-push barrier $\lambda$ in order
to properly balance between compression efficiency and update complexity.
Crucially, small storage can only be attained if the leaf-push barrier is
chosen according to \eqref{eq:leaf-push-barrier-H0}.  Strikingly, we found
that precisely this setting is the key to fast FIB updates as
well\footnote{Note that \eqref{eq:leaf-push-barrier-H0} transforms into
  \eqref{eq:leaf-push-barrier-delta} at maximum entropy.}.

Herein, we only specify the $\update$ operation that changes an
\emph{existing} association for prefix $a$ of prefix length $p$ to the new
label $s$ or, within the string model, rewrites an entire block of symbols at
the lowest level of the tree with a new one.  Adding a new entry or deleting
an existing one can be done in similar vein.

\vspace{.3em}\hrule height 1pt\vspace{.2em}%
  \begin{small}
    \begin{algorithmic}
      \algrenewcommand{\algorithmicfunction}{$\update$}

      \Function{}{address $a$, integer $p$, label $s$, integer $\lambda$}

      \State{$v \gets D(T).\Root$; $q \gets 0$}

      \While{$q < p$}

      \State{\textbf{if} $q \ge \lambda$ \textbf{then} $v \gets
        \decompress(v)$}

      \State{\textbf{if} $\getbits(a, q, 1) = 0$ \textbf{then} $v \gets
        v.\Left$ \textbf{else} $v \gets v.\Right$}

      \State{$q \gets q + 1$}

      \EndWhile

      \State{\textbf{if} $p < \lambda$ \textbf{then} $l(v) \gets s$; \textbf{return}}
      
      \State{$w \gets T.\Copy(v)$; re-pointer the parent of $v$ to $w$; $l(w) \gets s$}

      \State{\textsc{postorder-traverse-at-$v$}($u$: $\Get(u.\Left.\id,
        u.\Right.\id)$)}

      \State{$\triefold(w, 0)$}

      \State{\textbf{for each} parent $u$ of $w$: $\level(u) \ge \lambda$
        \textbf{do} $\compress(u)$}

      \EndFunction

      \algrenewcommand{\algorithmicfunction}{$\decompress$}

      \Function{}{node $v$}

      \State{$w \gets \New$;  $w.\id \gets v.\id$}

      \State{\textbf{if} $v \in L_D$ \textbf{then} $l(w) \gets l(v)$}

      \State{$\qquad$\textbf{else} $w.\Left \gets v.\Left$; $w.\Right \gets
        v.\Right$}

      \State{$\qquad\qquad \Get(v.\Left.\id, v.\Right.\id)$}

      \State{re-pointer the parent of $v$ to $w$; \textbf{return} $w$}

      \EndFunction
    \end{algorithmic}
  \end{small}
\hrule\vspace{.5em}%

First, we walk down and decompress the DAG along the path traced out by the
successive bits of $a$ until we reach level $p$.  The $\decompress$ routine
copies a node out from the DAG and removes the reference wherever necessary.
At this point, if $p<\lambda$ then we simply update the label and we are
ready. Otherwise, we replace the sub-trie below $v$ with a new copy of the
corresponding sub-trie from $T$, taking care of calling $\Get$ on the
descendants of $v$ to remove dangling references, and we set the label on the
root $w$ of the new copy to $s$.  Then, we re-compress the portions of the
prefix DAG affected by the change, by calling $\triefold$ on $w$ and then
calling $\compress$ on all the nodes along the upstream path from $w$ towards
to root.
\begin{theorem}
  If the leaf-push barrier $\lambda$ is set as
  \eqref{eq:leaf-push-barrier-H0}, then $\update$ on $D(T)$ terminates in
  $O(W(1+\frac{1}{H_0}))$ time.
\end{theorem}
\begin{proof}
  If $p<\lambda$, then updating a single entry can be done in $O(W)$ time.
  If, on the other hand, $p \ge \lambda$, then \texttt{update} visits at most
  $W + 2^{W-\lambda} \le W + \frac{W}{H_0}$ nodes, using that $\lambda \ge W
  - \lg( \tfrac{W}{H_0})$ whenever $\lambda$ is as
  \eqref{eq:leaf-push-barrier-H0}.
\end{proof}

In summary, under mild assumptions on the label distribution a prefix DAG
realizes the Shannon-entropy up to a small factor and allows indexing
arbitrary elements and updates to any entry in roughly $O(\log n)$ time.  As
such, it is in fact a dynamic, entropy-compressed string self-index.  As far
as we are aware of, this is the first pointer machine of this kind, as the
rest of compressed string-indexes are pointerless.  Regrettably, both the
space bound and the update complexity weaken when the label distribution is
extremely biased, i.e., when $H_0$ is very small.  As we argue in the next
section though, this rarely causes problems in practice.

\section{Numerical Evaluations}
\label{sec:num-eval}

At this point, we have yet to demonstrate that the appealing theoretical
properties of compressed FIBs indeed manifest as practical benefits.  For
this reason, we conducted a series of numerical evaluations with the goal to
quantify the compressibility of real IP FIBs and see how our compressors
fare\footnoteE{{\it Erratum:} All results have been updated to the correct
  FIB entropy bound and the revised \XBWL transform.}. It was \emph{not} our
intention, however, to compare to other FIB storage schemes from the
literature, let alone evince that ours is the fastest or the most efficient
one.  After all, information-theoretic space bounds are purposed precisely at
making such comparisons unnecessary, serving as analytically justified ground
truth.  Instead, our motivation is merely to demonstrate that FIB compression
allows to reduce memory tax without any compromise on the efficiency of
longest prefix match or FIB updates.

For the evaluations, we coded up a full-fledged Linux prototype, where FIB
compression and update run in user space and IP lookup is performed by a
custom kernel module embedded in the kernel's IP stack.  The code executed on
a single core of a 2.50GHz Intel Core i5 CPU, with 2x32 Kbyte L1 data cache,
256 Kbyte L2 cache, and 3 Mbyte L3 cache.

Research on IP FIB data structures has for a long time been plagued by the
unavailability of real data, especially from the Internet core.  Alas, we
could obtain only $5$ FIB instances from real IP routers, each from the
access: \texttt{taz} and \texttt{hbone} are from a university access,
\texttt{access(d)} is from a default and \texttt{access(v)} from a virtual
instance of a service provider's router, and \texttt{mobile} is from a mobile
operator's access (see Table~\ref{tab:static}).  The first $3$ are in the
DFZ, the rest contain default routes.  Apart from these, however, what is
available publicly is RIB dumps from BGP collectors, like RouteViews or
looking glass servers (named \texttt{as*} in the data set).  Unfortunately,
these only very crudely model real FIBs, because collectors run the BGP
best-path selection algorithm on their peers and these adjacencies differ
greatly from real next hops on production routers.  We experimented with
heuristics to restore the original next-hop information (e.g., set next-hop
to the first AS-hop), but the results were basically the same.  Thus, these
FIBs are included in the data set only for
reference. 
We also used two randomly generated FIBs, one of 600,000 (\texttt{fib\_600k})
and another of 1 million prefixes (\texttt{fib\_1m}), to future-proof our
results.  These synthetic FIBs were generated by iterative random prefix
splitting and setting next-hops according to a truncated Poisson-distribution
with parameter $\tfrac{3}{5}$ ($H_0=1.06$, $\delta=4$).

\subsection{Update Complexity}
\label{sec:leaf-push-choose}

First, we set out to determine a good setting for the leaf-push barrier
$\lambda$.  Recall that $\lambda$ was introduced to balance between the
compression efficiency and update complexity (also recall that no such
compromise exists between compression and \emph{lookup}.).  Our theoretical
results provide the essential pointers to set $\lambda$ (see
\eqref{eq:leaf-push-barrier-delta} and \eqref{eq:leaf-push-barrier-H0}), but
these are for compressing strings over complete binary trees.  IP FIBs,
however, are not complete.

We exercised the memory footprint vs. update complexity trade-off by varying
$\lambda$ between $0$ and $32$.  The update time was measured over two update
sequences: a random one with IP prefixes uniformly distributed on
$[0,2^{32}-1]$ and prefix lengths on $[0,32]$, and a BGP-inspired one
corresponding to a real BGP router log taken from RouteViews.  Here, we
treated all BGP prefix announcements as generating a FIB update, with a
next-hop selected randomly according to the next-hop distribution of the FIB.
The results are mean values over $15$ runs of $7,500$ updates, each run
roughly corresponding to $15$ minutes worth of BGP churn.

Herein, we only show the results for the \texttt{taz} FIB instance in
Fig.~\ref{fig:memsize-vs-update-taz}.  The results suggest that standard
prefix trees (reproduced by the setting $\lambda=32$), while pretty fast to
update, occupy a lot of space.  Fully compressed DAGs ($\lambda=0$), in
contrast, consume an order of magnitude less space but are expensive to
modify.  There is a domain, however, at around $5\le\lambda\le 12$, where we
win essentially all the space reduction and still handle about $100,000$
updates per second (that's roughly two and a half hours of BGP update load).
What is more, \emph{the space-time trade-off only exists for the synthetic,
  random update sequence, but not for BGP updates}.  This is because BGP
updates are heavily biased towards longer prefixes (with a mean prefix length
of $21.87$), which implies that the size of leaf-pushed sub-tries needed to
be re-packed per update is usually very small, and hence update complexity is
relatively insensitive to $\lambda$.

\begin{figure}
  \centering
\begin{tikzpicture}[gnuplot,scale=1.4]
\gpsolidlines
\gpmonochromelines
\gpcolor{color=gp lt color border}
\gpsetlinetype{gp lt border}
\gpsetlinewidth{1.00}
\draw[gp path] (0.000,0.000)--(0.180,0.000);
\draw[gp path] (4.061,0.000)--(3.881,0.000);
\node[gp node right] at (-0.184,0.000) {0.1};
\draw[gp path] (0.000,0.595)--(0.180,0.595);
\draw[gp path] (4.061,0.595)--(3.881,0.595);
\node[gp node right] at (-0.184,0.595) {1};
\draw[gp path] (0.000,1.190)--(0.180,1.190);
\draw[gp path] (4.061,1.190)--(3.881,1.190);
\node[gp node right] at (-0.184,1.190) {10};
\draw[gp path] (0.000,1.785)--(0.180,1.785);
\draw[gp path] (4.061,1.785)--(3.881,1.785);
\node[gp node right] at (-0.184,1.785) {100};
\draw[gp path] (0.000,2.380)--(0.180,2.380);
\draw[gp path] (4.061,2.380)--(3.881,2.380);
\node[gp node right] at (-0.184,2.380) {1K};
\draw[gp path] (0.000,2.975)--(0.180,2.975);
\draw[gp path] (4.061,2.975)--(3.881,2.975);
\node[gp node right] at (-0.184,2.975) {10K};
\draw[gp path] (0.101,0.000)--(0.101,0.180);
\draw[gp path] (0.101,3.391)--(0.101,3.211);
\node[gp node center] at (0.101,-0.308) {130K$\qquad$};
\draw[gp path] (0.644,0.000)--(0.644,0.180);
\draw[gp path] (0.644,3.391)--(0.644,3.211);
\node[gp node center] at (0.644,-0.308) {200K};
\draw[gp path] (1.800,0.000)--(1.800,0.180);
\draw[gp path] (1.800,3.391)--(1.800,3.211);
\node[gp node center] at (1.800,-0.308) {500K};
\draw[gp path] (2.675,0.000)--(2.675,0.180);
\draw[gp path] (2.675,3.391)--(2.675,3.211);
\node[gp node center] at (2.675,-0.308) {1M};
\draw[gp path] (3.549,0.000)--(3.549,0.180);
\draw[gp path] (3.549,3.391)--(3.549,3.211);
\node[gp node center] at (3.549,-0.308) {2M};
\draw[gp path] (0.000,3.391)--(0.000,0.000)--(4.061,0.000)--(4.061,3.391)--cycle;
\node[gp node center,rotate=-270] at (-0.890,1.695) {Update [$\mu$sec]};
\node[gp node left] at (3.831,-0.271) {[byte]};
\node[gp node left] at (0.557,3.154) {$\lambda=0$};
\node[gp node left] at (0.557,1.453) {$\lambda=0$};
\node[gp node left] at (0.635,0.456) {$\lambda=11$};
\node[gp node left] at (3.021,1.028) {$\lambda=32$};
\node[gp node right] at (3.509,3.043) {random};
\gpcolor{color=gp lt color 0}
\gpsetlinetype{gp lt plot 0}
\draw[gp path] (0.414,3.267)--(0.414,3.148)--(0.414,3.070)--(0.414,3.053)--(0.414,2.942)%
  --(0.414,2.800)--(0.413,2.508)--(0.413,2.220)--(0.411,1.919)--(0.442,1.612)--(0.496,1.417)%
  --(0.501,1.226)--(0.528,1.065)--(0.581,0.970)--(0.648,0.905)--(0.736,0.862)--(0.941,0.802)%
  --(1.226,0.773)--(1.383,0.752)--(1.675,0.703)--(1.995,0.684)--(2.447,0.764)--(2.858,0.772)%
  --(3.291,0.793)--(3.623,0.783)--(3.936,0.548)--(3.936,0.552)--(3.936,0.549)--(3.936,0.546)%
  --(3.936,0.504)--(3.937,0.518)--(3.937,0.522);
\gpsetpointsize{6.00}
\gppoint{gp mark 1}{(0.414,3.267)}
\gppoint{gp mark 1}{(0.414,3.148)}
\gppoint{gp mark 1}{(0.414,3.070)}
\gppoint{gp mark 1}{(0.414,3.053)}
\gppoint{gp mark 1}{(0.414,2.942)}
\gppoint{gp mark 1}{(0.414,2.800)}
\gppoint{gp mark 1}{(0.413,2.508)}
\gppoint{gp mark 1}{(0.413,2.220)}
\gppoint{gp mark 1}{(0.411,1.919)}
\gppoint{gp mark 1}{(0.442,1.612)}
\gppoint{gp mark 1}{(0.496,1.417)}
\gppoint{gp mark 1}{(0.501,1.226)}
\gppoint{gp mark 1}{(0.528,1.065)}
\gppoint{gp mark 1}{(0.581,0.970)}
\gppoint{gp mark 1}{(0.648,0.905)}
\gppoint{gp mark 1}{(0.736,0.862)}
\gppoint{gp mark 1}{(0.941,0.802)}
\gppoint{gp mark 1}{(1.226,0.773)}
\gppoint{gp mark 1}{(1.383,0.752)}
\gppoint{gp mark 1}{(1.675,0.703)}
\gppoint{gp mark 1}{(1.995,0.684)}
\gppoint{gp mark 1}{(2.447,0.764)}
\gppoint{gp mark 1}{(2.858,0.772)}
\gppoint{gp mark 1}{(3.291,0.793)}
\gppoint{gp mark 1}{(3.623,0.783)}
\gppoint{gp mark 1}{(3.936,0.548)}
\gppoint{gp mark 1}{(3.936,0.552)}
\gppoint{gp mark 1}{(3.936,0.549)}
\gppoint{gp mark 1}{(3.936,0.546)}
\gppoint{gp mark 1}{(3.936,0.546)}
\gppoint{gp mark 1}{(3.936,0.504)}
\gppoint{gp mark 1}{(3.937,0.518)}
\gppoint{gp mark 1}{(3.937,0.522)}
\gppoint{gp mark 1}{(3.693,3.043)}
\gpcolor{color=gp lt color border}
\node[gp node right] at (3.509,2.706) {BGP};
\gpcolor{color=gp lt color 1}
\gpsetlinetype{gp lt plot 1}
\draw[gp path] (0.414,1.382)--(0.414,1.310)--(0.414,1.306)--(0.414,1.378)--(0.414,1.375)%
  --(0.414,1.372)--(0.413,1.240)--(0.413,1.247)--(0.411,1.260)--(0.442,1.264)--(0.496,1.254)%
  --(0.501,1.232)--(0.528,1.201)--(0.581,1.168)--(0.648,1.132)--(0.736,1.045)--(0.941,0.990)%
  --(1.226,0.894)--(1.383,0.842)--(1.675,0.781)--(1.995,0.581)--(2.447,0.526)--(2.858,0.398)%
  --(3.291,0.224)--(3.623,0.134)--(3.936,0.079)--(3.936,0.074)--(3.936,0.089)--(3.936,0.103)%
  --(3.936,0.089)--(3.936,0.083)--(3.937,0.085)--(3.937,0.096);
\gppoint{gp mark 2}{(0.414,1.382)}
\gppoint{gp mark 2}{(0.414,1.310)}
\gppoint{gp mark 2}{(0.414,1.306)}
\gppoint{gp mark 2}{(0.414,1.378)}
\gppoint{gp mark 2}{(0.414,1.375)}
\gppoint{gp mark 2}{(0.414,1.372)}
\gppoint{gp mark 2}{(0.413,1.240)}
\gppoint{gp mark 2}{(0.413,1.247)}
\gppoint{gp mark 2}{(0.411,1.260)}
\gppoint{gp mark 2}{(0.442,1.264)}
\gppoint{gp mark 2}{(0.496,1.254)}
\gppoint{gp mark 2}{(0.501,1.232)}
\gppoint{gp mark 2}{(0.528,1.201)}
\gppoint{gp mark 2}{(0.581,1.168)}
\gppoint{gp mark 2}{(0.648,1.132)}
\gppoint{gp mark 2}{(0.736,1.045)}
\gppoint{gp mark 2}{(0.941,0.990)}
\gppoint{gp mark 2}{(1.226,0.894)}
\gppoint{gp mark 2}{(1.383,0.842)}
\gppoint{gp mark 2}{(1.675,0.781)}
\gppoint{gp mark 2}{(1.995,0.581)}
\gppoint{gp mark 2}{(2.447,0.526)}
\gppoint{gp mark 2}{(2.858,0.398)}
\gppoint{gp mark 2}{(3.291,0.224)}
\gppoint{gp mark 2}{(3.623,0.134)}
\gppoint{gp mark 2}{(3.936,0.079)}
\gppoint{gp mark 2}{(3.936,0.074)}
\gppoint{gp mark 2}{(3.936,0.089)}
\gppoint{gp mark 2}{(3.936,0.103)}
\gppoint{gp mark 2}{(3.936,0.089)}
\gppoint{gp mark 2}{(3.936,0.083)}
\gppoint{gp mark 2}{(3.937,0.085)}
\gppoint{gp mark 2}{(3.937,0.096)}
\gppoint{gp mark 2}{(3.693,2.706)}
\gpcolor{color=gp lt color border}
\gpsetlinetype{gp lt border}
\draw[gp path] (0.000,3.391)--(0.000,0.000)--(4.061,0.000)--(4.061,3.391)--cycle;
\gpdefrectangularnode{gp plot 1}{\pgfpoint{0.000cm}{0.000cm}}{\pgfpoint{4.061cm}{3.391cm}}
\end{tikzpicture}
  \caption{Update time vs. memory footprint on \texttt{taz} for random and
    BGP update sequences.}
  \label{fig:memsize-vs-update-taz}
\end{figure}
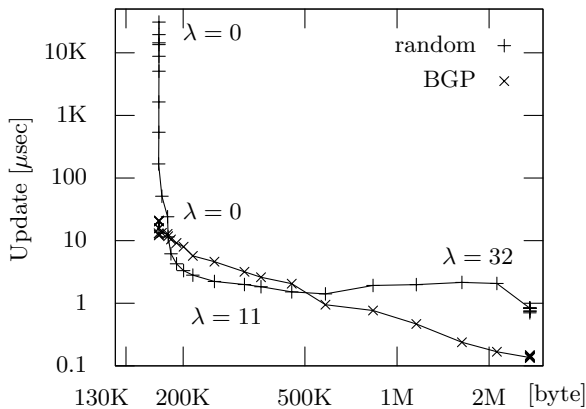

Based on these considerations, we set $\lambda=11$ for the rest of the
evaluations.

\begin{figure}
  \centering
\begin{tikzpicture}[gnuplot,scale=1.4]
\gpsolidlines
\gpmonochromelines
\gpcolor{color=gp lt color border}
\gpsetlinetype{gp lt border}
\gpsetlinewidth{1.00}
\draw[gp path] (0.000,0.754)--(0.180,0.754);
\node[gp node right] at (-0.184,0.754) {50};
\draw[gp path] (0.000,1.507)--(0.180,1.507);
\node[gp node right] at (-0.184,1.507) {100};
\draw[gp path] (0.000,2.261)--(0.180,2.261);
\node[gp node right] at (-0.184,2.261) {150};
\draw[gp path] (0.000,3.014)--(0.180,3.014);
\node[gp node right] at (-0.184,3.014) {200};
\draw[gp path] (0.000,0.000)--(0.000,0.180);
\draw[gp path] (0.000,3.391)--(0.000,3.211);
\node[gp node center] at (0.000,-0.308) {0.005};
\draw[gp path] (0.577,0.000)--(0.577,0.180);
\draw[gp path] (0.577,3.391)--(0.577,3.211);
\node[gp node center] at (0.577,-0.308) {0.01};
\draw[gp path] (1.154,0.000)--(1.154,0.180);
\draw[gp path] (1.154,3.391)--(1.154,3.211);
\node[gp node center] at (1.154,-0.308) {0.02};
\draw[gp path] (1.916,0.000)--(1.916,0.180);
\draw[gp path] (1.916,3.391)--(1.916,3.211);
\node[gp node center] at (1.916,-0.308) {0.05};
\draw[gp path] (2.493,0.000)--(2.493,0.180);
\draw[gp path] (2.493,3.391)--(2.493,3.211);
\node[gp node center] at (2.493,-0.308) {0.1};
\draw[gp path] (3.070,0.000)--(3.070,0.180);
\draw[gp path] (3.070,3.391)--(3.070,3.211);
\node[gp node center] at (3.070,-0.308) {0.2};
\draw[gp path] (3.833,0.000)--(3.833,0.180);
\draw[gp path] (3.833,3.391)--(3.833,3.211);
\node[gp node center] at (3.833,-0.308) {0.5};
\draw[gp path] (3.833,0.283)--(3.653,0.283);
\node[gp node left] at (4.017,0.283) { 1};
\draw[gp path] (3.833,0.848)--(3.653,0.848);
\node[gp node left] at (4.017,0.848) { 2};
\draw[gp path] (3.833,1.413)--(3.653,1.413);
\node[gp node left] at (4.017,1.413) { 3};
\draw[gp path] (3.833,1.978)--(3.653,1.978);
\node[gp node left] at (4.017,1.978) { 4};
\draw[gp path] (3.833,2.543)--(3.653,2.543);
\node[gp node left] at (4.017,2.543) { 5};
\draw[gp path] (3.833,3.108)--(3.653,3.108);
\node[gp node left] at (4.017,3.108) { 6};
\draw[gp path] (0.000,3.391)--(0.000,0.000)--(3.833,0.000)--(3.833,3.391)--cycle;
\node[gp node center,rotate=-270] at (-1.074,1.695) {Storage size [Kbytes]};
\node[gp node center,rotate=-270] at (4.722,1.695) {Compression efficiency};
\node[gp node right] at (2.100,3.031) {$H_0$};
\gpcolor{color=gp lt color 0}
\gpsetlinetype{gp lt plot 0}
\draw[gp path] (0.000,0.056)--(0.577,0.095)--(0.914,0.131)--(1.340,0.200)--(1.916,0.331)%
  --(2.254,0.441)--(2.493,0.532)--(2.831,0.685)--(3.070,0.802)--(3.256,0.897)--(3.408,0.964)%
  --(3.536,1.016)--(3.647,1.054)--(3.745,1.073)--(3.833,1.081);
\gpsetpointsize{8.00}
\gppoint{gp mark 1}{(0.000,0.056)}
\gppoint{gp mark 1}{(0.577,0.095)}
\gppoint{gp mark 1}{(0.914,0.131)}
\gppoint{gp mark 1}{(1.340,0.200)}
\gppoint{gp mark 1}{(1.916,0.331)}
\gppoint{gp mark 1}{(2.254,0.441)}
\gppoint{gp mark 1}{(2.493,0.532)}
\gppoint{gp mark 1}{(2.831,0.685)}
\gppoint{gp mark 1}{(3.070,0.802)}
\gppoint{gp mark 1}{(3.256,0.897)}
\gppoint{gp mark 1}{(3.408,0.964)}
\gppoint{gp mark 1}{(3.536,1.016)}
\gppoint{gp mark 1}{(3.647,1.054)}
\gppoint{gp mark 1}{(3.745,1.073)}
\gppoint{gp mark 1}{(3.833,1.081)}
\gppoint{gp mark 1}{(2.284,3.031)}
\gpcolor{color=gp lt color border}
\node[gp node right] at (2.100,2.671) {xbwb};
\gpcolor{color=gp lt color 1}
\gpsetlinetype{gp lt plot 1}
\draw[gp path] (0.000,0.075)--(0.577,0.123)--(0.914,0.166)--(1.340,0.245)--(1.916,0.388)%
  --(2.254,0.502)--(2.493,0.593)--(2.831,0.738)--(3.070,0.843)--(3.256,0.925)--(3.408,0.981)%
  --(3.536,1.023)--(3.647,1.053)--(3.745,1.068)--(3.833,1.074);
\gppoint{gp mark 2}{(0.000,0.075)}
\gppoint{gp mark 2}{(0.577,0.123)}
\gppoint{gp mark 2}{(0.914,0.166)}
\gppoint{gp mark 2}{(1.340,0.245)}
\gppoint{gp mark 2}{(1.916,0.388)}
\gppoint{gp mark 2}{(2.254,0.502)}
\gppoint{gp mark 2}{(2.493,0.593)}
\gppoint{gp mark 2}{(2.831,0.738)}
\gppoint{gp mark 2}{(3.070,0.843)}
\gppoint{gp mark 2}{(3.256,0.925)}
\gppoint{gp mark 2}{(3.408,0.981)}
\gppoint{gp mark 2}{(3.536,1.023)}
\gppoint{gp mark 2}{(3.647,1.053)}
\gppoint{gp mark 2}{(3.745,1.068)}
\gppoint{gp mark 2}{(3.833,1.074)}
\gppoint{gp mark 2}{(2.284,2.671)}
\gpcolor{color=gp lt color border}
\node[gp node right] at (2.100,2.311) {pDAG};
\gpcolor{color=gp lt color 2}
\gpsetlinetype{gp lt plot 2}
\draw[gp path] (0.000,0.193)--(0.577,0.292)--(0.914,0.368)--(1.340,0.586)--(1.916,1.007)%
  --(2.254,1.308)--(2.493,1.546)--(2.831,2.125)--(3.070,2.458)--(3.256,2.732)--(3.408,2.936)%
  --(3.536,3.091)--(3.647,3.199)--(3.745,3.269)--(3.833,3.285);
\gppoint{gp mark 3}{(0.000,0.193)}
\gppoint{gp mark 3}{(0.577,0.292)}
\gppoint{gp mark 3}{(0.914,0.368)}
\gppoint{gp mark 3}{(1.340,0.586)}
\gppoint{gp mark 3}{(1.916,1.007)}
\gppoint{gp mark 3}{(2.254,1.308)}
\gppoint{gp mark 3}{(2.493,1.546)}
\gppoint{gp mark 3}{(2.831,2.125)}
\gppoint{gp mark 3}{(3.070,2.458)}
\gppoint{gp mark 3}{(3.256,2.732)}
\gppoint{gp mark 3}{(3.408,2.936)}
\gppoint{gp mark 3}{(3.536,3.091)}
\gppoint{gp mark 3}{(3.647,3.199)}
\gppoint{gp mark 3}{(3.745,3.269)}
\gppoint{gp mark 3}{(3.833,3.285)}
\gppoint{gp mark 3}{(2.284,2.311)}
\gpcolor{color=gp lt color border}
\node[gp node right] at (2.100,1.951) {$\nu$};
\gpcolor{color=gp lt color 3}
\gpsetlinetype{gp lt plot 3}
\draw[gp path] (0.000,1.674)--(0.577,1.462)--(0.914,1.297)--(1.340,1.377)--(1.916,1.436)%
  --(2.254,1.394)--(2.493,1.360)--(2.831,1.470)--(3.070,1.449)--(3.256,1.439)--(3.408,1.439)%
  --(3.536,1.436)--(3.647,1.433)--(3.745,1.439)--(3.833,1.434);
\gppoint{gp mark 4}{(0.000,1.674)}
\gppoint{gp mark 4}{(0.577,1.462)}
\gppoint{gp mark 4}{(0.914,1.297)}
\gppoint{gp mark 4}{(1.340,1.377)}
\gppoint{gp mark 4}{(1.916,1.436)}
\gppoint{gp mark 4}{(2.254,1.394)}
\gppoint{gp mark 4}{(2.493,1.360)}
\gppoint{gp mark 4}{(2.831,1.470)}
\gppoint{gp mark 4}{(3.070,1.449)}
\gppoint{gp mark 4}{(3.256,1.439)}
\gppoint{gp mark 4}{(3.408,1.439)}
\gppoint{gp mark 4}{(3.536,1.436)}
\gppoint{gp mark 4}{(3.647,1.433)}
\gppoint{gp mark 4}{(3.745,1.439)}
\gppoint{gp mark 4}{(3.833,1.434)}
\gppoint{gp mark 4}{(2.284,1.951)}
\gpcolor{color=gp lt color border}
\gpsetlinetype{gp lt border}
\draw[gp path] (0.000,3.391)--(0.000,0.000)--(3.833,0.000)--(3.833,3.391)--cycle;
\gpdefrectangularnode{gp plot 1}{\pgfpoint{0.000cm}{0.000cm}}{\pgfpoint{3.833cm}{3.391cm}}
\end{tikzpicture}
  \caption{Size and compression efficiency $\nu$ over FIBs with
    Bernoulli distributed next-hops as the function of parameter $p$.}
  \label{fig:bernoulli-fib}
\end{figure}
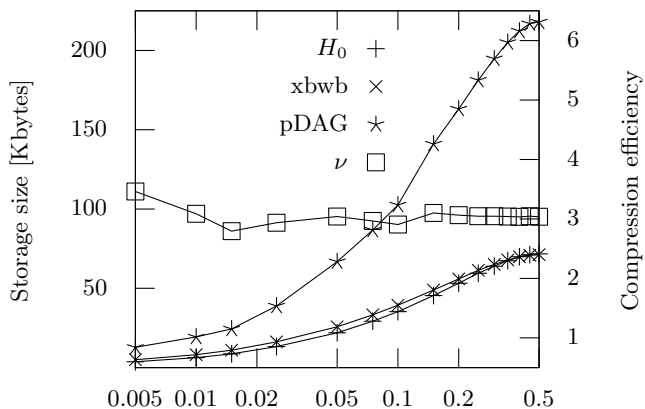

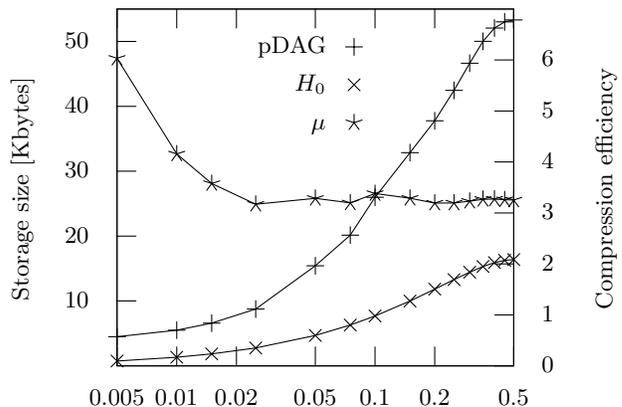
\begin{figure}
  \centering
\begin{tikzpicture}[gnuplot,scale=1.4]
\gpsolidlines
\gpmonochromelines
\gpcolor{color=gp lt color border}
\gpsetlinetype{gp lt border}
\gpsetlinewidth{1.00}
\draw[gp path] (0.000,0.617)--(0.180,0.617);
\node[gp node right] at (-0.184,0.617) {10};
\draw[gp path] (0.000,1.233)--(0.180,1.233);
\node[gp node right] at (-0.184,1.233) {20};
\draw[gp path] (0.000,1.850)--(0.180,1.850);
\node[gp node right] at (-0.184,1.850) {30};
\draw[gp path] (0.000,2.466)--(0.180,2.466);
\node[gp node right] at (-0.184,2.466) {40};
\draw[gp path] (0.000,3.083)--(0.180,3.083);
\node[gp node right] at (-0.184,3.083) {50};
\draw[gp path] (0.000,0.000)--(0.000,0.180);
\draw[gp path] (0.000,3.391)--(0.000,3.211);
\node[gp node center] at (0.000,-0.308) {0.005};
\draw[gp path] (0.567,0.000)--(0.567,0.180);
\draw[gp path] (0.567,3.391)--(0.567,3.211);
\node[gp node center] at (0.567,-0.308) {0.01};
\draw[gp path] (1.134,0.000)--(1.134,0.180);
\draw[gp path] (1.134,3.391)--(1.134,3.211);
\node[gp node center] at (1.134,-0.308) {0.02};
\draw[gp path] (1.883,0.000)--(1.883,0.180);
\draw[gp path] (1.883,3.391)--(1.883,3.211);
\node[gp node center] at (1.883,-0.308) {0.05};
\draw[gp path] (2.450,0.000)--(2.450,0.180);
\draw[gp path] (2.450,3.391)--(2.450,3.211);
\node[gp node center] at (2.450,-0.308) {0.1};
\draw[gp path] (3.017,0.000)--(3.017,0.180);
\draw[gp path] (3.017,3.391)--(3.017,3.211);
\node[gp node center] at (3.017,-0.308) {0.2};
\draw[gp path] (3.767,0.000)--(3.767,0.180);
\draw[gp path] (3.767,3.391)--(3.767,3.211);
\node[gp node center] at (3.767,-0.308) {0.5};
\draw[gp path] (3.767,0.000)--(3.587,0.000);
\node[gp node left] at (3.951,0.000) { 0};
\draw[gp path] (3.767,0.484)--(3.587,0.484);
\node[gp node left] at (3.951,0.484) { 1};
\draw[gp path] (3.767,0.969)--(3.587,0.969);
\node[gp node left] at (3.951,0.969) { 2};
\draw[gp path] (3.767,1.453)--(3.587,1.453);
\node[gp node left] at (3.951,1.453) { 3};
\draw[gp path] (3.767,1.938)--(3.587,1.938);
\node[gp node left] at (3.951,1.938) { 4};
\draw[gp path] (3.767,2.422)--(3.587,2.422);
\node[gp node left] at (3.951,2.422) { 5};
\draw[gp path] (3.767,2.907)--(3.587,2.907);
\node[gp node left] at (3.951,2.907) { 6};
\draw[gp path] (0.000,3.391)--(0.000,0.000)--(3.767,0.000)--(3.767,3.391)--cycle;
\node[gp node center,rotate=-270] at (-0.890,1.695) {Storage size [Kbytes]};
\node[gp node center,rotate=-270] at (4.656,1.695) {Compression efficiency};
\node[gp node right] at (2.067,3.031) {pDAG};
\gpcolor{color=gp lt color 0}
\gpsetlinetype{gp lt plot 0}
\draw[gp path] (0.000,0.277)--(0.567,0.339)--(0.899,0.406)--(1.317,0.540)--(1.883,0.951)%
  --(2.215,1.242)--(2.450,1.602)--(2.782,2.026)--(3.017,2.329)--(3.200,2.618)--(3.349,2.876)%
  --(3.475,3.084)--(3.584,3.210)--(3.681,3.269)--(3.767,3.287);
\gpsetpointsize{8.00}
\gppoint{gp mark 1}{(0.000,0.277)}
\gppoint{gp mark 1}{(0.567,0.339)}
\gppoint{gp mark 1}{(0.899,0.406)}
\gppoint{gp mark 1}{(1.317,0.540)}
\gppoint{gp mark 1}{(1.883,0.951)}
\gppoint{gp mark 1}{(2.215,1.242)}
\gppoint{gp mark 1}{(2.450,1.602)}
\gppoint{gp mark 1}{(2.782,2.026)}
\gppoint{gp mark 1}{(3.017,2.329)}
\gppoint{gp mark 1}{(3.200,2.618)}
\gppoint{gp mark 1}{(3.349,2.876)}
\gppoint{gp mark 1}{(3.475,3.084)}
\gppoint{gp mark 1}{(3.584,3.210)}
\gppoint{gp mark 1}{(3.681,3.269)}
\gppoint{gp mark 1}{(3.767,3.287)}
\gppoint{gp mark 1}{(2.251,3.031)}
\gpcolor{color=gp lt color border}
\node[gp node right] at (2.067,2.671) {$H_0$};
\gpcolor{color=gp lt color 1}
\gpsetlinetype{gp lt plot 1}
\draw[gp path] (0.000,0.046)--(0.567,0.082)--(0.899,0.114)--(1.317,0.170)--(1.883,0.289)%
  --(2.215,0.388)--(2.450,0.474)--(2.782,0.616)--(3.017,0.729)--(3.200,0.820)--(3.349,0.890)%
  --(3.475,0.944)--(3.584,0.981)--(3.681,1.003)--(3.767,1.010);
\gppoint{gp mark 2}{(0.000,0.046)}
\gppoint{gp mark 2}{(0.567,0.082)}
\gppoint{gp mark 2}{(0.899,0.114)}
\gppoint{gp mark 2}{(1.317,0.170)}
\gppoint{gp mark 2}{(1.883,0.289)}
\gppoint{gp mark 2}{(2.215,0.388)}
\gppoint{gp mark 2}{(2.450,0.474)}
\gppoint{gp mark 2}{(2.782,0.616)}
\gppoint{gp mark 2}{(3.017,0.729)}
\gppoint{gp mark 2}{(3.200,0.820)}
\gppoint{gp mark 2}{(3.349,0.890)}
\gppoint{gp mark 2}{(3.475,0.944)}
\gppoint{gp mark 2}{(3.584,0.981)}
\gppoint{gp mark 2}{(3.681,1.003)}
\gppoint{gp mark 2}{(3.767,1.010)}
\gppoint{gp mark 2}{(2.251,2.671)}
\gpcolor{color=gp lt color border}
\node[gp node right] at (2.067,2.311) {$\mu$};
\gpcolor{color=gp lt color 2}
\gpsetlinetype{gp lt plot 2}
\draw[gp path] (0.000,2.923)--(0.567,2.014)--(0.899,1.733)--(1.317,1.536)--(1.883,1.592)%
  --(2.215,1.549)--(2.450,1.638)--(2.782,1.593)--(3.017,1.547)--(3.200,1.548)--(3.349,1.565)%
  --(3.475,1.583)--(3.584,1.585)--(3.681,1.579)--(3.767,1.576);
\gppoint{gp mark 3}{(0.000,2.923)}
\gppoint{gp mark 3}{(0.567,2.014)}
\gppoint{gp mark 3}{(0.899,1.733)}
\gppoint{gp mark 3}{(1.317,1.536)}
\gppoint{gp mark 3}{(1.883,1.592)}
\gppoint{gp mark 3}{(2.215,1.549)}
\gppoint{gp mark 3}{(2.450,1.638)}
\gppoint{gp mark 3}{(2.782,1.593)}
\gppoint{gp mark 3}{(3.017,1.547)}
\gppoint{gp mark 3}{(3.200,1.548)}
\gppoint{gp mark 3}{(3.349,1.565)}
\gppoint{gp mark 3}{(3.475,1.583)}
\gppoint{gp mark 3}{(3.584,1.585)}
\gppoint{gp mark 3}{(3.681,1.579)}
\gppoint{gp mark 3}{(3.767,1.576)}
\gppoint{gp mark 3}{(2.251,2.311)}
\gpcolor{color=gp lt color border}
\gpsetlinetype{gp lt border}
\draw[gp path] (0.000,3.391)--(0.000,0.000)--(3.767,0.000)--(3.767,3.391)--cycle;
\gpdefrectangularnode{gp plot 1}{\pgfpoint{0.000cm}{0.000cm}}{\pgfpoint{3.767cm}{3.391cm}}
\end{tikzpicture}
  \caption{Size and compression efficiency $\nu$ over strings with Bernoulli
    distributed symbols as the function of parameter $p$.}
  \label{fig:bernoulli-string}
\end{figure}

\subsection{Storage Size}
\label{sec:num-entropy}

\begin{table*}
    \centering
    \caption{Results for $\XBWL$ and trie-folding on \emph{access},
      \emph{core}, and synthetic (\emph{syn.}) FIBs: name, number of prefixes
      $N$ and next-hops $\delta$; Shannon-entropy of the next-hop
      distribution $H_0$; FIB information-theoretic limit $I$, entropy $E$,
      and $\XBWL$ and prefix DAG size ($\pDAG$, $\lambda=11$) in KBytes;
      compression efficiency $\nu$; and bits/prefix efficiency for $\XBWL$
      ($\eta_{\XBWL}$) and trie-folding ($\eta_{\pDAG}$).}
    \label{tab:static}
    \begin{small}
      \renewcommand{\tabcolsep}{2.5pt}
      \renewcommand{\arraystretch}{1}
      \begin{tabular}{|c|l|ccc|cc|cc|c|cc|}
      \hline
      &FIB & $N$ & $\delta$ & $H_0$ & $I$ & $E$ & {\tiny $\XBWL$} & {\tiny $\pDAG$} & $\nu$ &
      {\tiny $\eta_{\tiny \XBWL}$} & {\tiny $\eta_{\pDAG}$}\\
      \hline
      \multirow{5}{.3em}{\rotatebox{90}{access}}
      & \texttt{taz}         & 410,513 & 4  & 1.00& 94 &  56&  63& 178&3.17&1.12 &  3.47\\
      & \texttt{hbone}       & 410,454 & 195& 2.00& 356& 142& 149& 396&2.78&1.05 &  7.71\\
      & \texttt{access(d)}   & 444,513 & 28 & 1.06& 206&  90& 100& 370&4.1 &1.12 &  6.65\\
      & \texttt{access(v)}   & 2,986   & 3  & 1.22& 2.8& 2.2& 2.5& 7.5&3.4 &1.13 & 20.23\\
      & \texttt{mobile}      & 21,783  & 16 & 1.08& 0.8& 0.4& 1.1& 3.6&8.71&2.36 &  1.35\\
      \hline
      \multirow{4}{.3em}{\rotatebox{90}{core}}
      & \texttt{as1221}      & 440,060 & 3  & 1.54&130 &115 &111 & 331&2.86 & 2.03 &  6.02\\
      & \texttt{as4637}      & 219,581 & 3  & 1.12& 52 & 41 & 44 & 129&3.13 & 1.62 &  4.69\\
      & \texttt{as6447}      & 445,016 & 36 & 3.91&375 &277 &277 & 748&2.7  &  5   & 13.45\\
      & \texttt{as6730}      & 437,378 & 186& 2.98&421 &209 &213 & 545&2.6  & 3.91 &  9.96\\
      \hline
      \multirow{2}{.3em}{\rotatebox{90}{syn.}}
      & \texttt{fib\_600k}   & 600,000  & 5  & 1.06&257 &157 &179 & 462&2.93 & 1.14 &  6.16\\
      & \texttt{fib\_1m}     & 1,000,000& 5  & 1.06&427 &261 &297 & 782&2.99 & 1.14 &  6.26\\
      \hline
    \end{tabular}
    \renewcommand{\arraystretch}{1}
  \end{small}
\end{table*}

Storage size results are given in Table~\ref{tab:static}.  Notably, real FIBs
that contain only a few next-hops compress down to about $60$--$150$ Kbytes
with $\XBWL$ at $1$--$2$ bit\slash prefix(!)  efficiency, and only about
$2$--$3$ times more with trie-folding.  This is chiefly attributed to the
small next-hop entropy, indicating the presence of a dominant next-hop.  Core
FIBs, on the other hand, exhibit considerably larger next-hop entropy, with
$\XBWL$ transforms in the range of $100$--$300$ and prefix DAGs in
$330$--$700$ KBytes.  Recall, however, that these FIBs exhibit unrealistic
next-hop distribution.  Curiously, even the extremely large FIB of 1 million
prefixes shrinks below $300$ Kbytes (800 KBytes with trie-folding).  In
contrast, small instances compress poorly, as it is usual in data
compression.  Finally, we observe that many \emph{FIBs show high next-hop
  regularity} (especially the real ones), reflected in the fact that entropy
bounds are $20$--$40$\% smaller than the information-theoretic limit.
$\XBWL$ \emph{very closely matches entropy bounds, with trie-folding off by
  only a small factor}.

We also studied compression ratios on synthetic FIBs, whose entropy was
controlled by us.  In particular, we re-generated the next-hops in
\texttt{access(d)} according to Bernoulli-distribution: a first next-hop was
set with probability $p$ and another with probability $1-p$.  Then, varying
$p$ in $[0, \tfrac{1}{2}]$ we observed the FIB entropy, the size of the
prefix DAG, and the compression efficiency $\nu$, i.e., the factor between
the two (see Fig.~\ref{fig:bernoulli-fib}). We found that \emph{the
  efficiency is around $3$ and}, in line with our theoretical analysis,
\emph{degrades as the next-hop distribution becomes extremely biased}.  This,
however, never occurs in reality (see again Table~\ref{tab:static}).  We
repeated the analysis in the string compression model: here, the FIB was
generated as a complete binary trie with a string of $2^{17}$ symbols written
on the leaves, again chosen by a Bernoulli distribution, and this was then
compressed with trie-folding (see Fig.~\ref{fig:bernoulli-string}, with
$\XBWL$ omitted). The observations are similar, with compression efficiency
again varying around $3$ and the spike at low entropy more
prominent\footnoteE{{\it Erratum:} Text updated to highlight that the
  compression efficiency in terms of the updated entropy measure has
  increased to $3$ for FIBs as well as for string compression, which is more
  in line with the theoretical bound.}.

\subsection{Lookup Complexity}
\label{sec:num-lookup}

\begin{table}
  \centering
  \caption{Lookup benchmark with $\XBWL$, prefix DAGs, \texttt{fib\_trie},
    and the FPGA implementation on \texttt{taz}: size, average and maximum
    depth; and million lookup per second, lookup time in CPU cycles, and
    cache misses per packet over random IP addresses (rand.) and addresses
    taken from the trace \cite{caida_trace} (trace).}
  \label{tab:lookup}
  \begin{small}
    \renewcommand{\tabcolsep}{2.5pt}
    \renewcommand{\arraystretch}{1}
    \begin{tabular}{|cl|ccc|c|}
      \hline
      && \multicolumn{3}{|c|}{Linux}  &  HW  \\
      \cline{3-6}
      && $\XBWL$  &  $\pDAG$ & \texttt{fib\_trie}  &  \texttt{FPGA}  \\
      \hline
      &\vline\hskip.3em size [Kbyte]      &     106  &     178  &    26,698  &  178   \\
      &\vline\hskip.3em average depth     &      --  &     3.7  &      2.42  &  --     \\
      &\vline\hskip.3em maximum depth     &      --  &      21  &         6  &  --     \\
      \hline
      \multirow{2}{.2em}{\rotatebox{90}{rand. }}
      &\vline\hskip.3em million lookup/sec &   0.033  &  12.8    &    3.23  & 6.9\\
      &\vline\hskip.3em CPU cycle/lookup   &   73940  &   194    &    771   & 7.1\\
      &\vline\hskip.3em cache miss/packet   &  0.016   &  0.003   &      3.17& --\\
      \hline
      \multirow{2}{.2em}{\rotatebox{90}{trace }}
      &\vline\hskip.3em million lookup/sec &   0.037  &  13,8    &     5.68 & 6.9\\
      &\vline\hskip.3em CPU cycle/lookup   &   67200  &  180     &     438  & 7.1\\
      &\vline\hskip.3em cache miss/packet   &   0.016  &  0.003  &      0.29 & --\\
      \hline
    \end{tabular}    
    \renewcommand{\arraystretch}{1}
  \end{small}
\end{table}

Finally, we tested IP lookup performance on real software and hardware
prototypes.  Our software implementations run inside the Linux kernel's IP
forwarding engine.  For this, we hijacked the kernel's network stack to send
IP lookup requests to our custom kernel module, working from a serialized
blob generated by the FIB encoders. Our $\XBWL$ lookup engine uses a kernel
port of the \texttt{RRR} bitstring index \cite{Raman:2002:SID:545381.545411}
and the Huffman-shaped \texttt{WaveletTree}
\cite{Ferragina:2007:CRS:1240233.1240243} from \texttt{libcds} \cite{libcds}.
Trie-folding was coded in pure \texttt{C}.  We used the standard trick to
collapse the first $\lambda=11$ levels of the prefix DAGs in the serialized
format \cite{Zec:2012:DTB:2378956.2378961}, as this greatly eases
implementation and improves lookup time with practically zero effect on
updates.  We also experimented with the Linux-kernel's stock
\texttt{fib\_trie} data structure, an adaptive level- and path-compressed
multibit trie-based FIB, as a reference implementation \cite{772439}.  Last,
we also realized the prefix DAG lookup algorithm in hardware, on a Xilinx
Virtex-II Pro 50 FPGA with 4.5 MBytes of synchronous SRAM representing the
state-of-the-art almost $10$ years ago.  The hardware implementation uses the
same serialized prefix DAG format as the software code.  All tests were run
on the \texttt{taz} instance.

For the software benchmarks we used the standard Linux network
micro-benchmark tool \texttt{kbench} \cite{nettesttools}, which calls the FIB
lookup function in a tight loop and measures the execution time with
nanosecond precision.  We modified \texttt{kbench} to take IP addresses from
a uniform distribution on $[0,2^{32}-1]$ or, alternatively, from a packet
trace in the ``CAIDA Anonymized Internet Traces 2012'' data set
\cite{caida_trace}. The route cache was disabled.
We also measured the rate of CPU cache misses by monitoring the
\texttt{cache-misses} CPU performance counter with the \texttt{perf(1)} tool.
For the hardware benchmark, we mirrored \texttt{kbench} functionality on the
FPGA, calling the lookup logic repeatedly on a list of IP addresses
statically stored in the SRAM and we measured the number of clock ticks
needed to terminate the test cycle. 

The results are given in Table~\ref{tab:lookup}. On the software side, the
most important observations are as follows.  The prefix DAG, taking only
about 180 KBytes of memory, is most of the time accessed from the cache,
while \texttt{fib\_trie} occupies an impressive $26$ MBytes and so it does
not fit into fast memory.  Thus, even though the number of memory accesses
needed to execute an IP lookup is smaller with \texttt{fib\_trie}, as most of
these go to slower memory the prefix DAG supports about three times as many
lookups per second.  Accordingly, \emph{not just that FIB space reduction
  does not ruin lookup performance, but it even improves it}. In other words,
there is no space-time trade-off involved here.  The address locality in real
IP traces helps \texttt{fib\_trie} performance to a great extent, as
\texttt{fib\_trie} can keep lookup paths to popular prefixes in cache.  In
contrast, the prefix DAG is pretty much insensitive to the distribution of
lookup keys.  Finally, we see that $\XBWL$ is a distant third from the tested
software lookup engines, suggesting that the constant in the lookup
complexity is indeed prohibitive in practice and that our lookup code
exercises some pretty pathologic code path in \texttt{libcds}.

The real potential of trie-folding is most apparent with our hardware
implementation.  The FPGA design executes a single IP lookup in just $7.1$
clock cycles on average, thanks to that the prefix DAG fits nicely into the
SRAM running synchronously with the logic.  This is enough to roughly $7$
million IP lookups per second even on our rather ancient FPGA board.  On a
modern FPGA or ASIC, however, with clock rates in the gigahertz range, our
results indicate that prefix DAGs could be scaled to hundreds of millions of
lookups per second at a terabit line speed.

We also measured packet throughput using the \texttt{udpflood}
macro-benchmark tool \cite{nettesttools}.  This tool injects UDP packets into
the kernel destined to a dummy network device, which makes it possible to run
benchmarks circumventing network device drivers completely.  The results were
similar as above, with prefix DAGs supporting consistently $2$--$3$ times
larger throughput than \texttt{fib\_trie}.

\section{Related Works}
\label{sec:related}

In line with the unprecedented growth of the routed Internet and the emerging
scalability concerns thereof \cite{potaroo, Zhao:2010:RSO:1878170.1878174,
  Khare:2010:ETG:1878170.1878183}, finding efficient FIB representations has
been a heavily researched question in the past and, judging from the
substantial body of recent work \cite{Han:2010:PGS:1851182.1851207,
  Uzmi:2011:SPN:2079296.2079325, Zec:2012:DTB:2378956.2378961,
  Liu:2012:EFC:2427036.2427039}, still does not seem to have been solved
completely.

Trie-based FIB schemes date back to the BSD kernel implementation of Patricia
trees \cite{Sklower91atree-based}. This representation consumes a massive 24
bytes per node, and a single IP lookup might cost $32$ random memory
accesses.  Storage space and search time can be saved on by expanding nodes'
strides to obtain a multibit trie \cite{752164}, see e.g., controlled prefix
expansion \cite{Srinivasan:1998:FIL:277858.277863,
  Ioannidis:2005:LCD:1114718.1648670}, level- and path-compressed tries
\cite{772439}, Lulea \cite{Degermark:1997:SFT:263105.263133}, Tree Bitmaps
\cite{Eatherton:2004:TBH:997150.997160} and successors
\cite{Song:2005:SST:1099544.1100365, Bando:2012:FBI:2369183.2369204}, etc.
Another approach is to shrink the routing table itself, by cleverly
relabeling the tree to contain the minimum number of entries (see ORTC and
derivatives \cite{draves:99, Uzmi:2011:SPN:2079296.2079325}). In our view,
trie-folding is complementary to these schemes, as it can be used in
combination with basically any trie-based FIB representation, realizing
varying extents of storage space reduction.

Further FIB representations include hash-based schemes
\cite{Waldvogel:1997:SHS:263105.263136, Bando:2012:FBI:2369183.2369204},
dynamic pipelining \cite{Hasan:2005:DPM:1080091.1080116}, CAMs \cite{253403},
Bloom-filters \cite{Dharmapurikar:2003:LPM:863955.863979}, binary search
trees and search algorithms \cite{DBLP:conf/infocom/GuptaPB00,
  Zec:2012:DTB:2378956.2378961}, massively parallelized lookup engines
\cite{Han:2010:PGS:1851182.1851207, Zec:2012:DTB:2378956.2378961}, FIB
caching \cite{Liu:2012:EFC:2427036.2427039}, and different combinations of
these (see the text book \cite{1200040}).  None of these come with
information-theoretic space bounds. Although next-hop entropy itself appears
in \cite{Uzmi:2011:SPN:2079296.2079325}, but no analytical evaluation ensued.
In contrary, $\XBWL$ and trie-folding come with \emph{theoretically proven}
space limits, and thus \emph{predicable memory footprint}.  The latest
reported FIB size bounds for >400K prefixes range from $780$ KBytes (DXR,
\cite{Zec:2012:DTB:2378956.2378961}) to $1.2$ Mbytes (SMALTA,
\cite{Uzmi:2011:SPN:2079296.2079325}).  $\XBWL$ improves this to just
$100$--$300$ Kbytes, which easily fits into today's SRAMs or can be realized
right in chip logic with modern FPGAs.

Compressed data structures have been in the forefront of theoretical computer
science research \cite{Hon:2010:CIR:1875737.1875761,
  Navarro:2007:CFI:1216370.1216372, Ferragina:2000:ODS:795666.796543, libcds,
  Makinen:2008:DES:1367064.1367072, Ziviani:2000:CKN:619057.621588,
  SilvadeMoura:2000:FFW:348751.348754, Raman:2002:SID:545381.545411,
  Ferragina:2007:CRS:1240233.1240243}, ever since Jacobson in his seminal
work \cite{63533} defined succinct encodings of trees that support
navigational queries in optimal time within information-theoretically limited
space.  Jacobson's bitmap-based techniques later found important use in FIB
aggregation \cite{Eatherton:2004:TBH:997150.997160,
  Song:2005:SST:1099544.1100365, Bando:2012:FBI:2369183.2369204}. With the
extensive use of bitmaps, $\XBWL$ can be seen as a radical rethinking of
these schemes, inspired by the state-of-the-art in succinct and compressed
data structures.

The basic idea of folding a labeled tree into a DAG is not particularly new;
in fact, this is the basis of many tree compacting schemes
\cite{KATAJAINEN:1990p202}, space-efficient ordered binary decision diagrams
and deterministic acyclic finite state automata
\cite{Bryant:1992:SBM:136035.136043}, common subexpression elimination in
optimizing compilers \cite{Cocke:1970:GCS:390013.808480}, and it has also
been used in FIB aggregation \cite{Ioannidis:2005:LCD:1114718.1648670,
  DBLP:conf/icnp/SongKHL09, 4694879} earlier.  Perhaps the closest to
trie-folding is Shape graphs \cite{DBLP:conf/icnp/SongKHL09}, where common
sub-trees, without regard to the labels, are merged into a DAG.  However,
this necessitates storing a giant hash for the next-hops, making updates
expensive especially considering that the underlying trie is leaf-pushed.
Trie-folding, in contrast, takes labels into account when merging and also
allows cheap updates.

\section{Conclusions}
\label{sec:conc}

With the rapid growth of the Web, social networks, mobile computing, data
centers, and the Internet routing ecosystem as a whole, the networking field
is in a sore need of compact and efficient data representations.  Today's
networking practice, however, still relies on ad-hoc and piecemeal data
structures for basically all storage sensitive and processing intensive
applications, of which the case of IP FIBs is just one salient example.

Our main goal in this paper was to advocate compressed data structures to the
networking community, pointing out that space reduction does not necessarily
hurt performance. Just the contrary: the smaller the space the more data can
be squeezed into fast memory, leading to faster processing.  This lack of
space-time trade-off is already exploited to a great extent in information
retrieval systems, business analytics, computational biology, and
computational geometry, and we believe that it is just too appealing not to
be embraced in networking as well. This paper is intended as a first step in
that direction, demonstrating the basic information-theoretic and algorithmic
techniques needed to attain entropy bounds, on the simple but relevant
example of IP FIBs.  Our techniques could then prove instructive in designing
compressed data structures for other large-scale data-intensive networking
applications, like OpenFlow and MPLS label tables, Ethernet self learning MAC
tables, BGP RIBs, access rules, log files, or peer-to-peer paths
\cite{madhyastha2009iplane}.

Accordingly, this paper can in no way be complete.  For instance, we
deliberately omitted IPv6 for brevity, even though storage burden for IPv6 is
getting just as pressing as for IPv4 \cite{Song:2005:SST:1099544.1100365}.
We see no reasons why our techniques could not be adapted to IPv6, but
exploring this area in depth is for further study.  Multibit prefix DAGs also
offer an intriguing future research direction, for their potential to reduce
storage space as well as improving lookup time from $O(W)$ to $O(\log W)$.
On a more theoretical front, FIB entropy lends itself as a new tool in
compact routing research, the study of the fundamental scalability of routing
algorithms.  We need to see why IP FIBs contain vast redundancy, track down
its origins and eliminate it, to enforce zero-order entropy bounds right at
the level of the routing architecture.  To what extent this argumentation can
then be extended to higher-order entropy is, for the moment, unclear at best.

\section*{Acknowledgements}
\label{sec:ack}

J.T. is with the MTA-Lendület Future Internet Research Group, and A. K. and
Z. H. are with the MTA-BME Information Systems Research Group.  The research
was partially supported by High Speed Networks Laboratory (HSN Lab),
J. T. was supported by the project TÁMOP - 4.2.2.B- 10/1--2010-0009, and G. R
by the OTKA/PD-104939 grant. The authors wish to thank Bence Mihálka, Zoltán
Csernátony, Gábor Barna, Lajos Rónyai, András Gulyás, Gábor Enyedi, András
Császár, Gergely Pongrácz, Francisco Claude, and Sergey Gorinsky for their
invaluable assistance, and to Jianyuan Lu \url{<lujy@foxmail.com>} for
pointing out the mistake regarding FIB entropy.

\vspace{.4em}
\begin{small}

\end{small}

\section*{Appendix}
\label{sec:proof_compact_data}

\begin{proof}[of Theorem~\ref{thm:proof_compact_data}]
  As $D(S)$ is derived from a complete binary tree the number of nodes
  $V^j_D$ of $D(S)$ at level $j$ is at most $\abs{V^j_D} \leq 2^j$, and each
  node at level $j$ corresponds to a $2^{W-j}$ long substring of $S$ so
  $\abs{V^j_D} \leq \delta^{2^{W-j}}$. Let $\kappa$ denote the intersection
  of the two bounds $2^{\kappa} = \delta^{2^{W-\kappa}}$, which gives:
  \begin{equation}\label{eq:kappa}
    \kappa 2^{\kappa} = 2^{W} \log_2(\delta) = n \log_2(\delta) \enspace .  
  \end{equation}
  Set the leaf push barrier at $\lambda=\lfloor \kappa \rfloor = \lfloor
  \frac1{\ln 2} \lambertW\left( n \ln \delta \right) \rfloor$ where
  $\lambertW()$ denotes the Lambert
  $W$-function. 
  The left side of Fig. \ref{fig:DAG-shape} is an illustration of the shape
  of DAG. Above level $\lambda$ we have
  $$\sum_{j=0}^{\lambda} \abs{V^j_D} = \sum_{j=0}^{\lambda} 2^j=
  2^{\lambda+1}-1\leq2\cdot2^{\kappa}$$ nodes; at level $\lambda+1$ we have
  $2^{\kappa}$ nodes at maximum; at $\lambda+2$ there are
  $\abs{V^{\lambda+2}_D} \leq \delta^{2^{W-\lambda-2}} \leq
  \sqrt{2^{\kappa}}$ nodes; and finally below level $\lambda+3$ we have an
  additional $\sqrt{2^{\kappa}}$ nodes at most as levels shrink as
  $\abs{V^{j+1}_D} \le \sqrt{\abs{V^{j}_D}}$ downwards in
  $D(S)$. 
  Finally, setting the pointer size at $\lceil \kappa \rceil$ bits and
  summing up the above yields that the size of $D(S)$ is at most
  \begin{multline*}
    \left(2 +2\left(1+ \frac{2}{\sqrt{2^{\kappa}}} \right)\right) \lceil
    \kappa
    \rceil 2^{\kappa} + \left(2\cdot2^\kappa+\delta\right)\log_2\delta\\
    =4 n \log_2(\delta) + o(n)
  \end{multline*}
  bits, using the fact that $\lceil \kappa \rceil 2^{\kappa}= n
  \log_2(\delta) + o(n)$ by \eqref{eq:kappa} and the number of labels stored
  in the DAG is at most $2\cdot2^\kappa$ above the leaf-push barrier and
  further $\delta$ below it.
\end{proof}

\begin{proof}[of Theorem~\ref{thm:proof_entropy}] 
  Let $E(|V^j_D|)$ denote the expected number of nodes at level $j$ of the
  DAG.  We shall use the following bounds on $E(|V^j_D|)$ to prove the claim:
  \begin{equation}\label{eq:vjd-ub}
    E(|V^j_D|) \leq \min\left\{2^{j}, \ \frac{H_0}{j} 2^{W} + 3,\ %
      \delta^{2^{W-j}} \right\} \enspace .
  \end{equation}

  Here, the first and the last bounds are from the proof of
  Theorem~\ref{thm:proof_compact_data}, while the second one is obtained
  below by treating the problem as a coupon collector's problem on the
  sub-tries of $D(T)$ at level $j$.  Suppose that we are given a set of
  coupons $C$, each coupon representing a string of length $2^{W-j}$ on the
  alphabet $\Sigma$ of size $\delta$ and entropy $H_0$, and we draw a coupon
  $o$ with probability $p_o: o \in C$.  Let $H_C = \sum_{o\in C} p_o \log_2
  \frac{1}{p_o} = H_0 2^{W-j}$, let $V$ denote the set of coupons after
  $m=2^{j}$ draws, and suppose $m\geq 3$.
  \begin{lemma}\label{lem:average_coupon}
    $E\left(\abs{V}\right) \leq \frac{m}{\log_2(m)}H_C+3$.
  \end{lemma}

  Using this Lemma, we have that the expected number of nodes at the $j$-th
  level of $D(S)$ is at most $E(\abs{V^j_D}) \leq \frac{2^{j}}{\log_2(2^{j})}
  H_0 2^{W-j} +3 = \frac{H_0}{j}n+3 = \frac{H_0}{j} 2^{W} + 3$, which
  coincides with the second bound in \eqref{eq:vjd-ub}.  Note that here $2^j$
  is an increasing function of $j$, while both $H_0 2^{W} / j + 3$ and
  $\delta^{2^{W-j}}$ are monotone decreasing functions.

  Using these bounds, we divide the DAG into three parts (the ``head'',
  ``body'', and ``tail'') as illustrated at the right side of
  Fig.~\ref{fig:DAG-shape}.  Let $\xi$ denote the intersection of the first
  two upper bounds, let $\zeta$ be that of the latter two and let $\kappa$ be
  the level where $2^j$ and $\delta^{2^{W-j}}$ meet.  It is easy too see,
  that the relation between these three values can only be $\xi \le
  \kappa\leq \zeta$ or $\xi \ge \kappa\ge \zeta$.  We discuss these two cases
  separately.

  \noindent\emph{Case 1:}  $\xi \le \kappa\leq \zeta$.
  
  The three parts of the DAG are as follows (again, see the right side of
  Fig.~\ref{fig:DAG-shape}):
  \begin{description}
  \item[\emph{head}] for levels $0,\dots,\lfloor\xi\rfloor$;
  \item[\emph{body}] for levels $\lfloor\xi\rfloor+1,\dots,\lceil\zeta\rceil-1$;
  \item[\emph{tail}] for levels $\lceil\zeta\rceil,\dots,W$.
  \end{description}

  In the following, we give upper bounds on the expected number of the nodes
  in the head, the tail, and the body of the DAG.  Set the leaf-push barrier
  at $\lambda=\lfloor \xi \rfloor$.
  
  \begin{figure}
    \small{
      \begin{tikzpicture}[thick, scale=0.10]
        \node at (52,92) {level $0$};
        \node at (52,39) {level $W$};
        \node at (101,48) {$\zeta$};
        \node at (26,50) {$\kappa$};
        \node at (101,60) {$\xi$};
        \node at (70.5,66) {$E(|V^j_D|) \leq 2^{j}$};
        \node at (71,55) {$E(|V^j_D|) \leq \frac{H_0}{j} 2^{W} + 3$};
        \node at (73.5,45) {$E(|V^j_D|) \leq \delta^{2^{W-j}}$};
        \node at (88,79.5) {head};
        \node at (88,51.5) {body};
        \node at (88,40.5) {tail};
        \draw (68,39) .. controls (61.23,44.41) and (51.30,47.72) .. (41.00,49.11) .. controls (37.70,49.56) and (34.38,49.80) .. (32,50);
        \draw (68,39) -- (72,39);
        \draw (72,39) .. controls (79.00,43.60) and (88.89,46.75) .. (95.5,48);
        \draw (96.5,60) -- (95.5,48);
        \draw[dashed] (28,50) -- (80,50);
        \draw[dotted] (60,60) -- (100,60);
        \draw[dotted] (60,48) -- (100,48);
        \draw (69.5,92) .. controls (60,69.5) and (50,60.5) .. (32.5,50);
        \draw (69.5,92) .. controls (81.5,69.5) and (86.5,65) .. (96.5,60);
      \end{tikzpicture}}
    \caption{The shape of the DAG as divided into three parts with bounds on
      the expected number of nodes at each level.}
    \label{fig:DAG-shape}
  \end{figure}
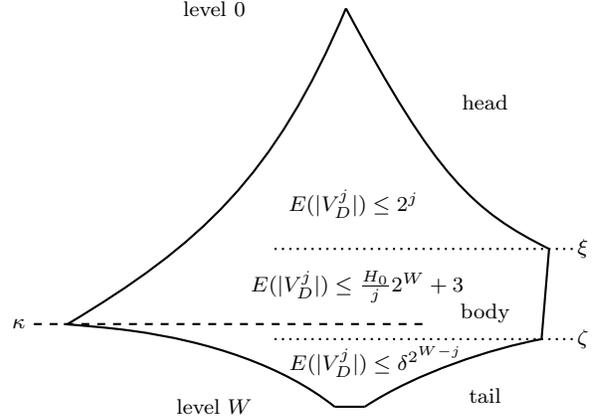
  
  First, the expected number of nodes in the head is
  \begin{multline}\label{size-head-node}
    \sum_{j=0}^{\lfloor\xi\rfloor}  E(|V^j_D|) \leq \sum_{j=0}^{\lfloor\xi\rfloor}  2^{j} = 2^{\lfloor\xi\rfloor+1}-1\\
    <2 \cdot 2^{\xi} = 2\left( \frac{H_0}{\xi} 2^{W} + 3\right) =2 M
  \end{multline}
  where $M = \left( \frac{H_0}{\xi} 2^{W} + 3\right)$.

  Second, for the size of the tail we have
  \begin{align}\label{size-tail-node}
    \sum_{j=\lceil\zeta\rceil}^{W}& E(|V^j_D|) \leq
    \sum_{j=\lceil\zeta\rceil}^{W} \delta^{2^{W-j}} =
    \sum_{j=1}^{W-\lceil\zeta\rceil} \delta^{2^{j}} \notag\\
    &=\delta^{2^{W-\lceil\zeta\rceil}} + \sum_{j=1}^{W-\lceil\zeta\rceil-1}
    \delta^{2^{j}}
    < \delta^{2^{W-\zeta}} + \sum_{i=1}^{2^{W-\lceil\zeta\rceil-1}}  \delta^{i}\notag\\
    &= \delta^{2^{W-\zeta}} + \delta^{2^{W-\lceil\zeta\rceil-1}+1}-2
    < \delta^{2^{W-\zeta}} + \delta \sqrt{\delta^{2^{W-\zeta}}}\notag\\
    &= \delta^{2^{W-\zeta}} \left(1 +
      \frac{\delta}{\sqrt{\delta^{2^{W-\zeta}}}}
    \right)=\left(\frac{H_0}{\zeta} 2^{W} + 3 \right) \left(1 + \epsilon_1
    \right)\notag\\
    &< \left(\frac{H_0}{\xi} 2^{W} + 3 \right) \left(1 + \epsilon_1 \right) =
    M(1+\epsilon_1) \enspace ,
  \end{align}
  where
  \begin{displaymath}
    \epsilon_1 = \frac{\delta}{\sqrt{\delta^{2^{W-\zeta}}}} \le
    \frac{\delta}{\sqrt{\delta^{\frac{W-\log_2\nicefrac{W}{H_0}}{\log_2\delta}}}}=\frac{\delta\sqrt{\nicefrac{W}{H_0}}}{2^{W/2}}
    \enspace , 
  \end{displaymath}
  which tends to zero if $W$ goes to infinity.

  Third, for the number of nodes in the body we write
  \begin{multline*}
    \sum_{j=\lfloor\xi\rfloor+1}^{\lceil\zeta\rceil-1} E(|V^j_D|) \leq
    \sum_{j=\lfloor\xi\rfloor+1}^{\lceil\zeta\rceil-1} \left(\frac{H_0}{j}
      2^{W} + 3 \right) <\\
    < \sum_{j=\lfloor\xi\rfloor+1}^{\lceil\zeta\rceil-1}
    \left(\frac{H_0}{\xi} 2^{W} + 3 \right)
    =(\lceil\zeta\rceil-1-\lfloor\xi\rfloor) M\\
    \le (\zeta-(\xi-1)) M = (1+\zeta-\xi) M \enspace .
  \end{multline*}

  \begin{lemma}\label{lem:bounds_on_levels}
    The following bounds on $\xi$, $\zeta$, and $\kappa$ apply:
    \begin{eqnarray}
      \xi &\geq & W-\log_2\nicefrac{W}{H_0} \enspace ,\label{eq:bound_on_xi}\\
      \zeta &\leq& W - \log_2 \left(W-\log_2 \nicefrac{W}{H_0} \right)+\log_2\log_2(\delta) \enspace ,\label{eq:bound_on_zeta}\\
      \kappa&\le&  W-\log_2\left(W-\log_2W\right)+\log_2\log_2\delta \enspace .\label{eq:bound_on_kappa}
    \end{eqnarray}
  \end{lemma}

  Using (\ref{eq:bound_on_xi}) and (\ref{eq:bound_on_zeta}) for the body we
  write
  \begin{align}\label{size-body-node}
    &\phantom{=}\sum_{j=\lfloor\xi\rfloor+1}^{\lceil\zeta\rceil-1}E(|V^j_D|) \leq(1+\zeta-\xi) M\notag\\
    &\le\left(1+ \log_2 \left( \frac{W}{H_0} \right) - \log_2
      \left({\frac{W-\log_2 \left( \frac{W}{H_0}
            \right)}{\log_2(\delta)}}\right) \right) M\notag\\
    &=\left(1 + \log\log_2(\delta)-\log_2H_0 + \epsilon_2\right) M,
  \end{align}
  where
  \begin{displaymath}
    \epsilon_2=\log_2  \frac{W}{W -\log_2 \left( \frac{W}{H_0} \right)}
    \xrightarrow[W\to\infty]{} 0 \enspace .
  \end{displaymath}

  Choose the pointer size to $\lceil\kappa\rceil$ bits, using that the DAG
  contains at most $2^{\lceil\kappa\rceil}$ nodes at its broadest level.  For
  the head we need one pointer for each node, while for the rest we need
  two. Summing up with \eqref{size-tail-node}, \eqref{size-head-node}, and
  \eqref{size-body-node} we get the following bound on the number of
  pointers:
  \begin{multline*}
    \left(2 + 2(1+\epsilon) +2\left(1-\log_2
        H_0+\log\log_2(\delta)+\epsilon_2\right)\right)M \\= \left(6+2\log_2
      \frac{\log_2\delta}{H_0}+2\left(\epsilon_1+\epsilon_2\right)\right)M
    \enspace .
  \end{multline*}

  We have to store labels above the barrier and at the bottom level, which is
  at most $2M+\delta$ labels, hence the average number of the bits is at most
  \begin{displaymath}
    \left(6+2\log_2
      \frac{\log_2\delta}{H_0}+2\left(\epsilon_1+\epsilon_2\right)\right)M\lceil\kappa\rceil+\left(2M+\delta\right)\log_2\delta
    \enspace . 
  \end{displaymath}

  Using \eqref{eq:bound_on_xi} for $\xi$ and \eqref{eq:bound_on_kappa} for
  $\kappa$ we have
  \begin{align*}
    &\phantom{=}M\lceil\kappa\rceil =
    \lceil\kappa\rceil\left(\frac{H_0}{\xi}2^W+3\right)\le
    \frac{\kappa+1}{\xi}H_0 n + 3\lceil\kappa\rceil\\
    &\le\left( \frac{W-\log_2\frac{W-\log_2 W}{\log_2\delta}+1}{W-\log_2\nicefrac{W}{H_0}}\right)H_0n + 3\lceil\kappa\rceil=\\
    &= \left(1+\frac{\log_2\nicefrac{W}{H_0}-\log_2\frac{W-\log_2
          W}{\log_2\delta}+1}{W-\log_2\nicefrac{W}{H_0}}\right)H_0n +
    3\lceil\kappa\rceil\\
    &=H_0n+o(n) \enspace.
  \end{align*}

  In summary, for the expected size of the DAG we get $(6-2\log_2 H_0 +
  2\log\log_2(\delta) )H_0 n+o(n)$ bits, since
  $\left(2M+\delta\right)\log_2\delta=o(n)$.

  \noindent\emph{Case 2:} $\zeta\le\kappa\le \xi$

  In this case the DAG contains only the head and tail parts. According to
  Theorem \ref{thm:proof_compact_data} we get the upper bound $5 n
  \log_2\delta+o(n)$ on the number of bits.  As $\kappa\le\xi$, we have
  \begin{displaymath}
    \frac{n\log_2\delta}{\kappa}=2^\kappa\le\frac{H_0}{\kappa}n+3\enspace.
  \end{displaymath}
  So the upper bound on the number of required bits is
  $5n\log_2\delta+o(n)\le 5 H_0 n+o(n) < (6 +
  2\log\left(\log_2(\delta)/H_0\right) )H_0 n+o(n)$.
\end{proof}

 \begin{proof}[of Lemma~\ref{lem:average_coupon}]
    The probability of having coupon $o$ in $V$ is $\Pr(o\in
    V)=1-(1-p_o)^{m}$ and so the expected cardinality of $V$ is $E(|V|)=
    \sum_{o\in C} (1-(1-p_o)^m)$.  The right-hand side of the statement of
    the Lemma is $\frac{m}{\log_2(m)} \sum_{o\in C} p_o \log_2 \frac{1}{p_o}
    +3$. The claim holds if $\forall o\in C$:
    \begin{equation}\label{eachp}
      p_o< \frac 1e \Rightarrow 1-(1-p_o)^m \leq \frac{m}{\log_2(m)} p_o \log_2 \frac{1}{p_o} \enspace .
    \end{equation}
    First, assume $m\geq\frac1{p_o}$. As the right hand size is a monotone
    increasing function of $m$ when $e \le \frac1{p_o}\le m$:
    \begin{small}
      \begin{equation*}
        1-(1-p_o)^m \leq 1 = \frac{1/p_o}{\log_2(1/p_o)} p_o \log_2
        \frac{1}{p_o} \leq \frac{m}{\log_2(m)} p_o \log_2 \frac{1}{p_o}
        \enspace .
      \end{equation*}
    \end{small}
    Otherwise, if $m<\frac1{p_o}$ then let $x=\log_{\frac1{p_o}} m$. Note
    that $0 < x < 1$.  After substituting $m=\frac1{p_o^x}$ we have
    \begin{small}
      \begin{multline*}
        1-(1-p_o)^{\frac1{p_o^x}} \leq
        \frac{\frac1{p_o^x}}{\log_2\left({\frac1{p_o^x}}\right)} p_o \log_2
        \left(\frac{1}{p_o}\right) =
         \\
         = \frac{\frac1{p_o^x}}{x
         \log_2\left({\frac1{p_o}}\right)} p_o \log_2
         \left(\frac{1}{p_o}\right) = \frac{\frac1{p_o^x}}{x} p_o =
        \frac1{xp_o^{x-1}} \enspace .
      \end{multline*}
    \end{small}
     which can be reordered as
     \begin{equation*} 
       (1-p_o)^{\frac1{p_o^x }}\geq 1 - \frac1{xp_o^{x-1}} \enspace.
     \end{equation*}
     Taking the $p_o^{x-1}>0$ power of both sides we get
     \begin{equation*} \label{reducedeq} (1-p_o)^{\frac1{p_o}}\geq \left({1 -
           \frac1{xp_o^{x-1}}}\right)^{p_o^{x-1}} .
     \end{equation*}
     Using that $x<1$ and so $\frac1x>1$, we see that the above holds if
     \begin{displaymath}
    (1-p_o)^{\frac1{p_o}}\geq \left({1 - \frac1{p_o^{x-1}}}\right)^{p_o^{x-1}}.
     \end{displaymath}
     Note that $(1-p_o)^{\frac1{p_o}}$ is monotone decreasing function, thus
     the inequality holds if $p_o \leq \frac1{p_o^{x-1}}$, but this is true
     because $p_o^x \leq 1$.  This proves \eqref{eachp} under the assumption
     $p_o < \frac{1}{e}$. Note also that there are at most $3>\frac1e$
     coupons for which \eqref{eachp} cannot be applied.
   \end{proof}

   \begin{proof}[of Lemma~\ref{lem:bounds_on_levels}]
     To prove \eqref{eq:bound_on_xi}, for level $l = W-\log_2
     \nicefrac{W}{H_0}$ we have
     \begin{equation*} \label{k-bounds}
       2^{l} = 2^{W-\log_2  \nicefrac{W}{H_0} } \leq 2^{W-\log_2 
         \nicefrac{W}{H_0}} = 2^W\frac{H_0}{W}  < \frac{H_0}{l} 2^{W} + 3
       \enspace .
     \end{equation*}
 
     For \eqref{eq:bound_on_zeta}, at level $l = W - \log_2
     \left(\frac{W-\log_2 \nicefrac{W}{H_0} }{\log_2\delta} \right)$ we have
     \begin{multline*} \label{l-bounds} \delta^{2^{W-l}} = \delta^{
         \frac{W-\log_2 \left( \frac{W}{H_0} \right)}{\log_2(\delta)} }=
       \delta^{\frac{\log_2 \left( \frac{H_0}{W} 2^W\right)}{\log_2(\delta)}
       }\\
       =\frac{H_0}{W} 2^W < \frac{H_0}{l} 2^{W} + 3 \enspace .
  \end{multline*}
  
  Finally, for \eqref{eq:bound_on_kappa} at level
  $l=W-\log_2\left(\frac{W-\log_2W}{\log_2\delta}\right)$ we write
  \begin{displaymath}
    \delta^{2^{W-l}}=\delta^\frac{W-\log_2W}{\log_2\delta}=2^{W-\log_2W}\le2^{l}
    \enspace . 
  \end{displaymath}

  This completes the proof.
\end{proof}

\end{document}